\documentclass[a4paper]{article}

\usepackage[utf8]{inputenc}
\usepackage[round]{natbib}
\usepackage[english]{babel}
\usepackage{amsfonts}
\usepackage{amsmath}
\usepackage{amsthm}
\usepackage{amssymb}
\usepackage{dsfont}
\usepackage{url}
\usepackage{hyperref}
\usepackage{tikz-cd}
\usepackage{mathtools}
\usepackage{enumitem}
\usepackage{color}
\usepackage{mathpartir}
\usepackage{stmaryrd}

\mathtoolsset{showonlyrefs}

\newtheorem{theorem}{Theorem}
\newtheorem{corollary}[theorem]{Corollary}

\newtheorem{proposition}[theorem]{Proposition}
\theoremstyle{remark}
\newtheorem{remark}[theorem]{Remark}
\theoremstyle{definition}
\newtheorem{definition}[theorem]{Definition}

\begin{document}

\title{Algebraic Semantics of Datalog with Equality}

\author{Martin E. Bidlingmaier}

\date{}

\maketitle

\begin{abstract}
  We discuss the syntax and semantics of relational Horn logic (RHL) and partial Horn logic (PHL).
  RHL is an extension of the Datalog programming language that allows introducing and equating variables in conclusions.
  PHL is a syntactic extension of RHL by partial functions and one of the many equivalent notions of essentially algebraic theory.

  Our main contribution is a new construction of free models.
  We associate to RHL and PHL sequents \emph{classifying morphisms}, which enable us to characterize logical satisfaction using lifting properties.
  We then obtain free and weakly free models using the \emph{small object argument}.
  The small object argument can be understood as an abstract generalization of Datalog evaluation.
  It underpins the implementation of the \emph{Eqlog} Datalog engine, which computes free models of PHL theories.
\end{abstract}

\section{Introduction}

\emph{Datalog} \citep{know-datalog} is a programming language for logical inference from Horn clauses.
Abstracting from concrete syntax, a Datalog program consists of the following declarations:
\begin{itemize}
  \item
    A set of \emph{sort symbols} $s$.
  \item
    A set of \emph{relation symbols} and their \emph{arities} $r : s_1 \times \dots \times s_n$.
  \item
    A set of \emph{sequents} (or \emph{rules}, or \emph{axioms}) of the form
    \begin{equation}
      r_1(\bar v^1) \land \dots \land r_n(\bar v^n) \implies r_{n + 1}(\bar v^{n+1})
    \end{equation}
    where $\bar v^i = (v_1, \dots, v_{k_i})$ is a sort-compatible list of variables for each $i$, and each variable in the conclusion also appears in the premise.
\end{itemize}
A \emph{fact} is an expression of the form $r(c_1, \dots, c_n)$ where each $c_i$ is a constant symbol of the appropriate sort.
Given a Datalog program and a set of input facts, a \emph{Datalog engine} computes the set of facts that can be derived from the input facts by repeated application of sequents.

A typical example of a problem that can be solved using Datalog is the computation of the transitive closure of a (directed) graph.
Graphs are given by a binary relation $E \subseteq V \times V$ of edges among a sort $V$ of vertices.
The only axiom of transitive graphs is the transitivity axiom
\begin{equation}
  E(u, v) \land E(v, w) \implies E(u, w).
\end{equation}
A set of input facts for this Datalog program is given by a set of expressions $E(a, b)$ where $a, b$ are constant symbols.
We identify such data with the data of a graph with vertices $V = \bigcup_{E(a, b)} \{ a, b \}$ and edges $E = \{(a, b) \mid E(a, b) \}$.
Every finite graph in which every vertex appears in some edge arises in this way, so we conflate such graphs and sets of facts.
(Standard Datalog does not support constants that do not appear in a fact.)

Given the Datalog program for transitive graphs and a corresponding set of facts, a Datalog engine enumerates all matches of the premise of the transitivity axiom, i.e. all substitutions $u \mapsto a, v \mapsto b, w \mapsto c$ such that the substituted conjuncts of the premise, $E(a, b)$ and $E(b, c)$, are in the set of input facts.
For each such substitution, the Datalog engine then adds the substitution $E(a, c)$ of the conclusion to the set of facts.
This process is repeated until the set of facts does not increase anymore; that is, until a fixed point has been reached.
This final set of facts now corresponds to a transitive graph.

Datalog has seen renewed interest in recent years for the implementation of program analysis tasks \citep{declarative-points-to, alias-analysis-bdd, flix} such as points-to analysis.
One encodes abstract syntax trees derived from the program source code as relations, on which one then runs Datalog programs.
The advantage of this approach over more ad-hoc methods is that implementation time can be reduced significantly, and that different analyses can be integrated seamlessly.

\emph{Equality saturation} has garnered interest as a program optimization technique in recent years \citep{egg}.
The idea is to insert expressions that should be optimized into an e-graph, and then close the e-graph under a set of rewrite rules.
E-graphs allow sharing nodes that occur as children more than once, so that a large number of expressions can be stored.
Furthermore, e-graphs can be efficiently closed under congruence, i.e. equivalence can be propagated from subexpressions to their parents.
After a suitable number of rewrite rules have been applied and the e-graph has been closed under congruence, one selects a suitable equivalent expression from the equivalence class of the expression one is interested in according to a cost function.
Crucially, equality saturation makes considerations about the order of rewrites unnecessary.

In this paper, we study languages and corresponding semantics that combine and subsume both Datalog and the applications of e-graphs outlined above.
To that end, we extend Datalog by \emph{equality}, that is, the ability of enforce an equality $u \equiv v$ in the conclusion of a sequent.
One example is the order-theoretic antisymmetry axiom
\begin{equation}
  \mathrm{Le}(u, v) \land \mathrm{Le}(v, u) \implies u \equiv v
\end{equation}
which is not valid Datalog due to the equality atom $u \equiv v$, but allowed in our extension.
If during evaluation of RHL an equality among constants $c_1$ and $c_2$ is inferred, we expect the system to conflate $c_1$ and $c_2$ in all contexts henceforth.
In other words, inferred equality should behave as congruence with respect to relations.
For example, the premise $E(u, v) \land E(v, w)$ of the transitivity axiom should match $(b, c_1), (c_2, d) \in E$ if the equality $c_1 \equiv c_2$ has been inferred earlier.
In addition to a set of derived facts, we also expect evaluation to yield an equivalence relation on each sort, representing inferred equalities.

\emph{Relational Horn logic} extends Datalog further by \emph{sort quantification}, i.e. variables matching any element of a sort, and by variables that only occur in a conclusion.
We interpret the latter as existentially quantified:
If the premise of a sequent matches and the conclusion contains a variable that is not bound in the premise, then we expect the Datalog engine to create new identifiers of the given sort if necessary to ensure that the conclusion holds.

\emph{Partial Horn logic}, originally due to \citet{phl}, is a layer of syntactic sugar on top RHL, i.e. a purely syntactic extension with the same descriptive power.
PHL adds function symbols $f : s_1 \times \dots s_n \rightarrow s$, which desugar into relations $f : s_1 \times \dots \times s_n \times s$ representing the graph of the function and the functionality axiom
\begin{equation}
  f(v_1, \dots, v_n, u) \land f(v_1, \dots, v_n, w) \implies u \equiv w.
\end{equation}
In positions where RHL expects variables (e.g. arguments of predicates or in equations), PHL allows also composed terms.
A composed term is desugared into a fresh variable corresponding to the result of applying the function and an assertion about the graph of the function.

These features enable the implementation of algorithms in PHL for which Datalog unsuitable, for example congruence closure \citep{congruence-closure}, Steensgaard's points-to analysis \citep{steensgaard} and type inference \citep[22.3, 22.4]{pierce-types-and-programming-languages}.
In each case, evaluation of the PHL theory encoding the problem domain yields the same algorithm as the standard domain-specific algorithm.
However, the present paper focuses on the semantics of PHL and RHL, whereas an evaluation algorithm and the applications above are presented in \citet{eqlog-algorithm}.

Partial Horn logic is one of the equivalent notions of \emph{essentially algebraic theory} \citep[Chapter 3.D]{locally-presentable-and-accessible-categories}.
Essentially algebraic theories generalize the better-known algebraic theories of universal algebra by allowing functions to be partial.
Crucially, the \emph{free model theorem} of universal algebra continues to hold also for essentially algebraic theories.
Free models are the basis of our semantics of PHL evaluation.
We show that free models can be computed using the \emph{small object argument}, which we shall come to understand as an abstract generalization of Datalog evaluation.

In brief, the relation of free models and Datalog evaluation can be understood for the transitivity Datalog program outlined above as follows.
We have seen that input data for this Datalog program represent certain graphs $G = (V, E)$, while output data represent transitive graphs $G' = (V, E')$.
The two graphs $G$ and $G'$ share the same set of vertices $V$, which is the set of constant symbols that appear in the set of input facts.
Intuitively, $G'$ arises from $G$ by adding data that must exist due to the transitivity axiom but no more.

Let us rephrase the relation between $G$ and $G'$ using category theory.
Denote by $\mathrm{Graph}$ the category of graphs:
A morphism $f : (V_1, E_1) \rightarrow (V_2, E_2)$ between graphs is a map $f : V_1 \rightarrow V_2$ that preserves the edge relation.
Thus if $(u, v) \in E_1$, then we must have $(f(u), f(v)) \in E_2$.
The requirement that the output graph $G'$ arises from the input $G$ solely by application of the transitivity sequent can now be summarized as follows:
\begin{proposition}
  \label{prop:free-transitive-graph}
  Let $G' = (V, E')$ be the output graph generated from evaluating the transitivity Datalog program on a finite input graph $G = (V, E)$.
  Then $G'$ is the free transitive graph over $G$.
\end{proposition}
\begin{proof}
  First we must exhibit a canonical graph morphism $\eta : G \rightarrow G'$.
  As $G$ and $G'$ share the same set of vertices, we choose $\eta$ simply as identity map on $V$.
  Note that the identity on $V$ is indeed a graph morphism $(V, E) \rightarrow (V, E')$ because $E \subseteq E'$.
  
  Now we must show that for all graph morphisms $f : G \rightarrow H$ where $H = (V_H, E_H)$ is a transitive graph, there exists a unique graph morphism $\bar f : G' \rightarrow H$ such that the following triangle commutes:
  \begin{equation}
    \begin{tikzcd}
      G \arrow[r, "f"] \arrow[d, "\eta"'] & H \\
      G' \arrow[ur, dashed, "\exists! \bar f"']
    \end{tikzcd}
  \end{equation}
  Because $\eta$ is the identity map, it suffices to show that $f = \bar f$ also defines a graph morphisms $G' \rightarrow H$; uniqueness of $\bar f$ follows from surjectivity of $\eta$.
  Recall that $G'$ arises from repeatedly matching the premise of the transitivity axiom and adjoining its conclusion.
  Thus there is a finite chain
  \begin{equation}
    E = E_0  \subseteq E_1 \subseteq \dots \subseteq E_n = E'
  \end{equation}
  where for each $i$ there exist $a, b, c \in V$ such that
  \begin{equation}
    \label{eq:trans-graph-step}
    E_{i + 1} = E_{i} \cup \{(a, c\}
    \qquad\qquad
    (a, b), (b, c) \in E_i.
  \end{equation}
  By induction, it suffices to show for all $i$ that $f$ is a graph morphism $(V, E_{i + 1}) \rightarrow H$ assuming that $f$ is a graph morphisms $(V, E_i) \rightarrow H$.
  Choose $a, b, c$ such that $(a, b), (b, c) \in E_i$ and \eqref{eq:trans-graph-step} is satisfied.
  Because $f$ is a graph morphism $(V, E_i) \rightarrow (V, E_{i + 1})$, we have $(f(a), f(b)), (f(b), f(c)) \in E_H$.
  Because $H$ is transitive, it follows that $(f(a), f(c)) \in E_H$.
  Thus $f$ preserves the edge $(a, c)$ and hence constitutes a graph morphism $(V, E_{i + 1}) \rightarrow H$.
\end{proof}
Denote by $\mathrm{TGraph}$ the full subcategory of $\mathrm{Graph}$ given by the transitive graphs.
The inclusion functor $\mathrm{TGraph} \subseteq \mathrm{Graph}$ has a left adjoint, a \emph{reflector}, which is given by assigning a graph to its transitive hull.
Thus Proposition~\ref{prop:free-transitive-graph} shows that the transitivity Datalog program computes the reflector.
Our primary goal in this paper is to explore and extend a semantics of PHL along these lines.

\textbf{Outline and Contributions.}
In Section~\ref{sec:small-object-argument}, we review the \emph{small object argument} \citep[Theorem 2.1.14]{hovey-model-categories} as a method of computing weak reflections into subcategories of injective objects.
We introduce \emph{strong} classes of morphisms, for which the small object argument specializes to the \emph{orthogonal-reflection construction} \citep[Chapter 1.C]{locally-presentable-and-accessible-categories} and produces a reflection into orthogonal subcategories.

In Section~\ref{sec:rhl}, we introduce \emph{relational Horn logic (RHL)}.
RHL extends Datalog with sort quantification, with variables that occur only in the conclusion, and with equations.
Input data of Datalog programs generalize to finite \emph{relational structures}, and output data generalize to \emph{models}, i.e. relational structures that satisfy all sequents.

Our poof of the existence of free or weakly free models associates to each RHL sequent a \emph{classifying morphism} of relational structures.
Satisfaction of the sequent can be characterized as lifting property against the classifying morphism.
The small object argument now shows the existence of weakly free models.
From this perspective, we may thus understand the small object argument as an abstract formulation of Datalog evaluation.

In Section~\ref{sec:phl}, we extend RHL by function symbols to obtain \emph{partial Horn logic (PHL)}.
By identifying each function symbol with a relation symbol representing its graph and adding a functionality axiom, every PHL theory gives rise to a relational Horn logic theory with equivalent semantics.
For \emph{epic} PHL theories, where all variables must be introduced in the premise of a sequent, the associated RHL theory is strong.
Conversely, we show that the semantics of every strong RHL theory can be recovered as semantics of an epic PHL theory.
This justifies the usage of epic PHL as an equally powerful but syntactically better-behaved language compared to strong RHL.

The results of this paper serve as semantics of \href{https://github.com/eqlog/eqlog}{\emph{Eqlog}}, a Datalog engine that computes free models of epic PHL theories.
Eqlog's algorithm is based on an efficient implementation of the small object argument that combines optimized Datalog evaluation (semi-naive evaluation and indices) with techniques used in congruence closure algorithms.
The present paper focuses on the semantics of PHL and RHL, whereas the evaluation algorithm employed by Eqlog is presented in \citet{eqlog-algorithm}.
Independently of Eqlog and the work presented there, members of the Egg \citep{egg} community have recently created the \emph{Egglog} tool, which combines Datalog with e-graphs and is based on very similar ideas as those of Eqlog.

\section{The Small Object Argument}
\label{sec:small-object-argument}

This section is a review of the small object argument, which we shall in later sections come to understand as an abstract description of Datalog evaluation.
The concepts we discuss here are not new and are in fact widely known among homotopy theorists; see for example \citet{hovey-model-categories} for a standard exposition.
A minor innovation is our consideration of \emph{strong} sets:
Sets of morphisms for which injectivity coincides with orthogonality.
For strong sets, the small object argument yields a reflection into the orthogonal subcategory where in general we would obtain only a weak reflection into the injective subcategory.

The \emph{orthogonal-reflection construction} \citep[Chapter 1.C]{locally-presentable-and-accessible-categories} produces a reflection into the orthogonal subcategory for arbitrary sets of morphisms $M$.
We show that every set of morphism $M$ can be extended to a strong set $N$ such that $M$ and $N$ induce the same orthogonality class.
The small object argument for $N$ now specializes to the orthogonal-reflection construction for $M$.
Thus, the concept of strong morphisms can be used to understand the orthogonal-reflection construction as a specialized variation of the small object argument.

Fix a cocomplete locally small category $\mathcal{C}$ for the remainder of this section.
We reserve the word \emph{set} for a small set, while \emph{class} refers to a set in a larger set-theoretic universe that contains the collection of objects in $\mathcal{C}$.
All colimits of set-indexed diagrams in $\mathcal{C}$ exist, while colimits of class-indexed diagrams need not exist.

\begin{definition}
  Let $f : A \rightarrow B$ be a morphism and let $X$ be an object.
  We write $f \pitchfork X$ and say that $X$ is \emph{injective} to $f$ if for all maps $a : A \rightarrow X$ there exists a map $b : B \rightarrow X$ such that
  \begin{equation}
    \begin{tikzcd}
      A \arrow[r, "a"] \arrow[d, "f"'] & X \\
      B \arrow[ur, dashed, "\exists b"']
    \end{tikzcd}
  \end{equation}
  commutes.
  If furthermore $b$ is unique for all $a$, then we write $f \perp X$ and say that $X$ is \emph{orthogonal} to $f$.

  If $M$ is a class of morphisms, then we write $M \pitchfork X$ if $f \pitchfork X$ for all $f \in M$, and $M \perp X$ if $f \perp X$ for all $f \in M$.
  The full subcategories given by the injective and orthogonal objects, respectively, are denoted by $M^\pitchfork$ and $M^\perp$.
  We call $M^\pitchfork$ the \emph{injectivity class} of $M$ and $M^\perp$ the \emph{orthogonality class} of $M$.
\end{definition}

\begin{definition}
  \label{def:strong}
  A class $M$ of morphisms is called \emph{strong} if $M^\pitchfork = M^\perp$.
\end{definition}

One of the main sources of strong sets is the following proposition:
\begin{proposition}
  \label{prop:epi-strong}
  Let $M$ be a class of epimorphisms.
  Then $M$ is strong.
  \qed
\end{proposition}
\begin{proof}
  This follows immediately from right-cancellation.
\end{proof}

Another source of strong sets is the following proposition.
It lets us reduce questions about orthogonality classes to strong injectivity classes.
\begin{proposition}
  \label{prop:strongbization}
	Let $M$ be a class of morphisms.
	Then there exists a superclass $N \supseteq M$ such that $N$ is strong and $N^\pitchfork = M^\perp$.
  If $M$ is a set, then $N$ can be chosen as set.
\end{proposition}
\begin{proof}
  Let $f : A \rightarrow B$ be a morphism in $M$.
  Then for each object $X$, the data of a single map $a : A \rightarrow X$ and two maps $b_1, b_2 : B \rightarrow X$ such that
  \begin{equation}
    \begin{tikzcd}
      A \arrow[r, "a"] \arrow[d, "f"'] & X \\
      B \arrow[ur, "b_i"'] &
    \end{tikzcd}
  \end{equation}
  commutes for $i \in \{1, 2\}$ is in bijective correspondence to a map $\langle b_1, b_2 \rangle : B \amalg_A B \rightarrow X$.
  Let
  \begin{equation}
    f' : B \amalg_A B \rightarrow B.
  \end{equation}
  be the canonical map that collapses the two copies of $B$ into one.
  Then $b_1 = b_2$ if and only if there exists a map $b$ such that
  \begin{equation}
    \begin{tikzcd}
      B \amalg_A B \arrow[r, "{\langle b_1, b_2\rangle}"] \arrow[d, "f'"'] & X \\
      B \arrow[ur, "b"'] &
    \end{tikzcd}
  \end{equation}
  commutes.
  The map $f'$ is an epimorphism.
  Thus if $b$ exists, then it exists uniquely, and $b = b_1 = b_2$.
  It follows that $X$ is orthogonal to $f$ if and only if $f$ is injective to both $f$ and $f'$.
  The desired class $N$ can thus be defined by $N = M \cup \{ f' \mid f \in M \}$.
\end{proof}

\begin{definition}
  A \emph{sequence} of morphisms is a diagram of the form
  \begin{equation}
    \begin{tikzcd}
      X_0 \arrow[r, "f_0"] & X_1 \arrow[r, "f_1"] & \dots
    \end{tikzcd}
  \end{equation}
  for a countable set $(f_n)_{n \in \mathbb{N}}$ of morphisms.
  The \emph{(infinite) composition} of a sequence of morphisms $(f_n)_{n \in \mathbb{N}}$ is the canonical map
  \begin{equation}
    X_0 \rightarrow X_\infty = \operatorname{colim}_{n \geq 0} X_n
  \end{equation}
  to the colimit of the sequence.
\end{definition}
Note that the composition of a sequence of morphisms is uniquely determined only up to a choice of colimit.

\begin{definition}
  \label{def:relative-cell-complex}
  Let $M$ be a class of morphisms.
  The class $\mathrm{Cell}(M)$ of \emph{relative $M$-cell complexes} is the least class of morphisms such that the following closure properties hold:
  \begin{enumerate}
    \item
      \label{itm:cell-base}
      $M \subseteq \mathrm{Cell}(M)$.
    \item
      \label{itm:cell-coproduct}
      $\mathrm{Cell}(M)$ is closed under coproducts.
      That is, if $(f_i : A_i \rightarrow B_i)_{i \in I}$ is a family of morphisms indexed by some set $I$ and $f_i \in \mathrm{Cell}(M)$ for all $i \in I$, then
      \begin{equation}
        \coprod_{i \in I} f_i : \coprod_{i \in I} A_i \rightarrow \coprod_{i \in I} B_i
      \end{equation}
      is in $\mathrm{Cell}(M)$.
    \item
      \label{itm:cell-pushout}
      $\mathrm{Cell}(M)$ is closed under pushouts.
      That is, if
      \begin{equation}
        \begin{tikzcd}
          A \arrow[d] \arrow[r, "f"] \arrow[dr, very near end, "\ulcorner", phantom] & B \arrow[d] \\
          X \arrow[r, "f'"] & Y
        \end{tikzcd}
      \end{equation}
      is a pushout square and $f \in \mathrm{Cell}(M)$, then $f' \in \mathrm{Cell}(M)$.
    \item
      \label{itm:cell-composition}
      $\mathrm{Cell}(M)$ is closed under composition of sequences.
      That is, if
      \begin{equation}
        \begin{tikzcd}
          A_0 \arrow[r, "f_0"] & A_1 \arrow[r, "f_1"] & \dots
        \end{tikzcd}
      \end{equation}
      is a sequence of morphisms $f_n \in \mathrm{Cell}(M)$ with composition $f : A_0 \rightarrow A_\infty$, then $f \in \mathrm{Cell}(M)$.
  \end{enumerate}
\end{definition}

\begin{remark}
  \label{rem:iterative-flattening}
  Standard literature on factorization systems and the closely related small object argument \citep{hovey-model-categories} usually considers not only countable sequences of morphisms but also arbitrary transfinite sequences, which are chains of morphisms indexed by an arbitrary ordinal number.
  In this more general setting, one typically defines a relative $M$-cell complex to be a transfinite composition of pushouts of morphisms in $M$ without mention of coproducts.

  This more common notion of relative $M$-cell complex satisfies our closure properties \ref{itm:cell-base} -- \ref{itm:cell-composition}.
  For \ref{itm:cell-coproduct}, one chooses a well-ordering on the indexing set $I$, and then computes the coproduct as composition of a chain indexed by this well-ordering.
  Conversely, our definition of relative $M$-cell complex is closed under arbitrary transfinite composition if all morphism in $M$ have finitely presentable domains and codomains (Definition~\ref{def:finitely-presentable}).
  Thus, whenever the domains and codomains of the morphisms in $M$ are finitely presentable, the definition given here and the usual one agree.
\end{remark}

\begin{proposition}
  \label{prop:iterative-relative-cell-complexes}
  Let $M$ be a class of morphisms.
  Define classes of morphisms $M \subseteq M_1 \subseteq M_2 \subseteq M_3$ as follows:
  \begin{mathpar}
    M_1 = M \cup \{ f \mid f \text{ is a coproduct of morphisms in } M \}
    \\
    M_2 = M_1 \cup \{ f \mid f \text{ is a pushout of a morphism in } M_1 \}
    \\
    M_3 = M_2 \cup \{ f \mid f \text{ is a composition of a sequence of morphisms in } M_2 \}
  \end{mathpar}
  Then $M_3 = \mathrm{Cell}(M)$.
\end{proposition}
\begin{proof}
  Coproducts, pushouts and compositions of sequences are all defined via colimits.
  Because colimits commute with colimits, $M_3$ is closed under coproducts, pushouts and compositions of sequences.
  It follows that $\mathrm{Cell}(M) \subseteq M_3$, hence $\mathrm{Cell}(M) = M_3$.
\end{proof}

\begin{proposition}
  Let $M$ be a class of morphisms.
  Then $\mathrm{Cell}(M)^\pitchfork = M^\pitchfork$ and $\mathrm{Cell}(M)^\perp = M^\perp$.
\end{proposition}
\begin{proof}
  If $M \supseteq N$ is an inclusion of classes of morphisms, then in general $M^\pitchfork \subseteq N^\pitchfork$ and $M^\perp \subseteq N^\perp$.
  This proves the inclusions $\subseteq$.

  Conversely, it suffices to show for $X \in M^\pitchfork$ that the class
  \begin{equation}
    N = \{f \in \operatorname{Mor} \mathcal{C} \mid f \pitchfork X \}
  \end{equation}
  satisfies the closure properties \ref{itm:cell-base} -- \ref{itm:cell-composition} of Definition~\ref{def:relative-cell-complex}, and similarly for orthogonality.
  This is routine.
  For example, closure under pushouts can be proved as follows.
  Let $X$ be injective to $f : A \rightarrow B$, let $f' : A' \rightarrow B'$ be a pushout of $f$, and let $A' \rightarrow X$ be an arbitrary morphism.
  The lift $B' \rightarrow X$ can then be obtained from a lift $B \rightarrow X$ and the universal property of the pushout as depicted in the following commuting diagram:
  \begin{equation}
    \begin{tikzcd}
      A \arrow[d, "f"] \arrow[r] & A' \arrow[r] \arrow[d, "f'"] & X \\
      B \arrow[r] \arrow[u] \arrow[rru, dashed, bend right=45] & B' \arrow[ur, dashed]
    \end{tikzcd}
  \end{equation}
  If the lift $B \rightarrow X$ is unique, then also $B' \rightarrow X$ is unique by uniqueness of the morphism induced by the universal property of the pushout.
\end{proof}

\begin{definition}
  \label{def:reflection}
  Let $\mathcal{C}' \subseteq \mathcal{C}$ be a full subcategory.
  A \emph{weak reflection} of an object $X \in \mathcal{C}$ into $\mathcal{C}'$ is a map $\eta : X \rightarrow X'$ such that $X' \in \mathcal{C}'$ and every map $X \rightarrow Y$ with $Y \in \mathcal{C}'$ factors via $\eta$.
  If the factorization is unique for all $X \rightarrow Y$, then $\eta$ is a \emph{reflection}.
  A \emph{(weak) reflector} consists of a functor $F : \mathcal{C} \rightarrow \mathcal{C}'$ and a natural transformation $\eta : \mathrm{Id} \rightarrow F$ such that $\eta_X$ is a (weak) reflection for all $X \in \mathcal{C}$.
  The subcategory $\mathcal{C}'$ is \emph{(weakly) reflective} in $\mathcal{C}$ if there exists a (weak) reflector.
\end{definition}

\begin{proposition}
  \label{prop:factoring-via-cell-complexes}
  Let $M$ be a class of morphisms.
  Let $f : X \rightarrow Y$ be a relative $M$-cell complex, and let $g : X \rightarrow Z$ be a map with $Z \in M^\pitchfork$.
  Then there is a map $h : Y \rightarrow Z$ such that $hf = g$.
  If furthermore $M$ is strong, then $h$ is unique.
\end{proposition}
\begin{proof}
  This follows from the fact that the class of morphisms $f$ for which the proposition holds satisfies properties \ref{itm:cell-base} -- \ref{itm:cell-composition} of Definition~\ref{def:relative-cell-complex}.
\end{proof}

\begin{proposition}
  \label{prop:relative-m-cell-complex-is-reflection}
  Let $M$ be a class of morphisms.
  Let $f : X \rightarrow Y$ be a relative $M$-cell complex such that $Y \in M^\pitchfork$.
  Then $f$ is a weak reflection into $M^\pitchfork$.
  If $M$ is strong, then $f$ is a reflection.
\end{proposition}
\begin{proof}
  By Proposition~\ref{prop:factoring-via-cell-complexes}.
\end{proof}

\begin{definition}
  \label{def:finitely-presentable}
  An object $X$ is \emph{finitely presentable} if the hom-functor $\mathrm{Hom}(X, -) : \mathcal{C} \rightarrow \mathrm{Set}$ preserves filtered colimits.
\end{definition}

\begin{proposition}[Small Object Argument: Property]
  \label{prop:small-object-argument-prop}
  Let $M$ be a class of morphisms with finitely presentable domains and codomains.
  Let
  \begin{equation}
    \begin{tikzcd}
      X_0 \arrow[r, "x_0"] & X_1 \arrow[r, "x_1"] & \dots
    \end{tikzcd}
  \end{equation}
  be a sequence in $M$ such that the following holds:
  \begin{enumerate}
    \item
      $x_n$ is a relative $M$-cell complex for all $n$.
    \item
      \label{itm:m-factorizations-exist}
      For all $f : A \rightarrow B$ in $M$, $n \geq 0$ and maps $a : A \rightarrow X_n$, there exists a map a map $b : B \rightarrow X_m$ for some $m \geq n$ such that
      \begin{equation}
        \begin{tikzcd}
          A \arrow[rrr, "f"] \arrow[d, "a"'] & & & B \arrow[d, "b"] \\
          X_n \arrow[r, "x_n"] & X_{n + 1} \arrow[r, "x_{n + 1}"] & \dots \arrow[r, "x_{m - 1}"] & X_m
        \end{tikzcd}
      \end{equation}
      commutes.
  \end{enumerate}
  Then the infinite composition $X_0 \rightarrow X_\infty$ of the $x_n$ is a weak reflection into $M^\pitchfork$.
  If $M$ is strong, then $X_0 \rightarrow X_\infty$ is a reflection.
\end{proposition}
\begin{proof}
  Because $\mathrm{Cell}(M)$ is closed under infinite composition, the map $X_0 \rightarrow X_\infty$ is a relative $M$-cell complex.
  Thus by Proposition~\ref{prop:relative-m-cell-complex-is-reflection}, it suffices to show that $X_\infty$ is in $M^\pitchfork$.

  Let $f : A \rightarrow B$ be in $M$ and let $a : A \rightarrow X_\infty$.
  Because $A$ is finitely presentable, there exists $n \in \mathbb{N}$ and $a_n : A \rightarrow X_n$ such that $a$ factors as $A \rightarrow X_n \rightarrow X_\infty$.
  By assumption \ref{itm:m-factorizations-exist}, there exist $m$ and $b_m : B \rightarrow X_m$ that commutes with $f$, $a_n$ and $x_{m - 1} \circ \dots \circ x_n$.
  Thus if we define $b$ as composition $B \rightarrow X_m \rightarrow X_\infty$, then $a = b \circ f$.
\end{proof}

\begin{proposition}[Small Object Argument: Existence]
  \label{prop:small-object-argument-exist}
  Let $M$ be a set of morphism with finitely presentable domains and codomains, and let $X$ be an object.
  Then there exists a sequence
  \begin{equation}
    \begin{tikzcd}
      X = X_0 \arrow[r, "x_0"] & X_1 \arrow[r, "x_1"] & \dots
    \end{tikzcd}
  \end{equation}
  satisfying the conditions of Proposition~\ref{prop:small-object-argument-prop}.
  In particular, $M^\pitchfork$ is a (weakly) reflective subcategory of $\mathcal{C}$.
\end{proposition}
\begin{proof}
  It suffices to construct a relative $M$-cell complexes $X \rightarrow Y$ such that for every $f : A \rightarrow B$ in $M$ and $a : A \rightarrow X$, there exists a commuting diagram
  \begin{equation}
    \begin{tikzcd}
      A \arrow[r, "f"] \arrow[d, "a"'] & B \arrow[d, "b"] \\
      X \arrow[r] & Y.
    \end{tikzcd}
  \end{equation}
  We then obtain the desired sequence by induction.

  Let $K$ be the set of pairs $(f, a)$, where $f : A \rightarrow B$ is a morphism in $M$ and $a : A \rightarrow X$.
  Note that $K$ is a set because $M$ is a set and $\mathrm{Hom}(A, X)$ is a set for all $A$.
  Now let $X \rightarrow Y$ be the map defined by the following pushout diagram:
  \begin{equation}
    \begin{tikzcd}
      \coprod_{(f, a) \in K} \operatorname{dom} f \arrow[r] \arrow[d] \arrow[dr, phantom, very near end, "\ulcorner"] & \coprod_{(f, a) \in K} \operatorname{cod} f \arrow[d] \\
      X \arrow[r] & Y
    \end{tikzcd}
  \end{equation}
  Here the top map is the coproduct $\coprod_{(f, a) \in K} f$, and the left vertical map $\langle a \rangle_{(f, a) \in K}$ is induced by the universal property of coproducts.
\end{proof}

\begin{proposition}
  \label{prop:reflection-vs-injectivity}
  Let $M$ be a strong set of morphisms, and let $f : A \rightarrow B$.
  Denote by $\bar f : \bar A \rightarrow \bar B$ the reflection of $f$ into $M^\pitchfork$.
  Then the following equations among injectivity and orthogonality classes hold:
  \begin{mathpar}
    (M \cup \{ f \})^\perp = (M \cup \{ \bar f \})^\perp
    \and
    (M \cup \{ f \})^\pitchfork = (M \cup \{ \bar f \})^\pitchfork.
  \end{mathpar}
\end{proposition}
\begin{proof}
  Let $X$ be orthogonal (equivalently: injective) to $M$.
  Then there is a bijective correspondence between solutions to the following lifting problems:
  \begin{equation}
    \begin{tikzcd}    
      A \arrow[r, "a"] \arrow[d, "f"'] & X \\
      B \arrow[ur, dashed]
    \end{tikzcd}    
    \qquad
    \begin{tikzcd}    
      \bar A \arrow[r, "a'"] \arrow[d, "\bar f"'] & X \\
      \bar B. \arrow[ur, dashed]
    \end{tikzcd}    
  \end{equation}
  Here $a$ is an arbitrary map, and $a' : \bar A \rightarrow X$ is induced from $a$ by the universal property of $\bar A$ and $X$ being orthogonal to $\bar f$.
\end{proof}

\section{Relational Horn Logic}
\label{sec:rhl}

Relational Horn Logic (RHL) is a superset of Datalog.
Most notably, RHL allows equations, and in particular equations in conclusions.
Our semantics of RHL are based on \emph{relational structures}, which we introduce in Section~\ref{subsec:relational-structures}.
In Section~\ref{subsec:rhl-semantics}, we then consider syntax and semantics of RHL.
We show that RHL models can be characterized using lifting properties against \emph{classifying morphisms}.
This enables us to apply the small object argument to prove the existence of (weakly) free models, in close analogy to Datalog evaluation.
In Section~\ref{subsec:rhl-completeness}, we prove a completeness result for the descriptive power of RHL:
Every finitary injectivity class of relational structures can be obtained as semantics of an RHL theory.
In Section~\ref{subsec:datalog-vs-rhl}, we identify in detail the subset of RHL that corresponds to Datalog.
We then explain how the computation of free RHL models can be reduced to evaluation of Datalog with minor extensions via the \emph{setoid transformation}.

\subsection{Relational Structures}
\label{subsec:relational-structures}

\begin{definition}
  A \emph{relational signature} $\mathfrak{S}$ is given by the following data:
  \begin{itemize}
    \item
      A set $S$ of \emph{sort symbols}.
    \item
      A set $R$ of \emph{relation symbols}.
    \item
      A map that assigns to each relation symbol $r \in R$ an \emph{arity}
      \begin{equation}
        r : s_1 \times \dots \times s_n
      \end{equation}
      of sort symbols $s_1, \dots, s_n \in S$ for $n \geq 0$.
  \end{itemize}
\end{definition}

\begin{definition}
  Let $\mathfrak{S} = (S, R)$ be a relational signature.
  A \emph{relational $\mathfrak{S}$-structure} consists of the following data:
  \begin{itemize}
    \item
      For each sort symbol $s \in S$, a \emph{carrier set} $X_s$.
    \item
      For each relation symbol $r \in R$ with arity $r : s_1 \times \dots \times s_n$, a relation $r_X \subseteq X_{s_1} \times \dots \times X_{s_n}$.
  \end{itemize}

  A \emph{morphism of relational structures} $f : X \rightarrow Y$  consists of functions $f_s : X_s \rightarrow Y_s$ for $s \in S$ that are compatible with the relations $r_X$ and $r_Y$ for all $r$.
  That is, we require that if $(x_1, \dots, x_n) \in r_X$ for some relation symbol $r : s_1 \times \dots \times s_n$, then $(f_{s_1}(x_1), \dots, f_{s_n}(x_n)) \in r_Y$.
  The category of relational structures is denoted by $\mathrm{Rel}(\mathfrak{S})$.
\end{definition}

When no confusion can arise, we suppress sort annotations.
Thus if $X$ is a relational structure, then we write $x \in X$ to mean that $x \in X_s$ for some $s \in S$.
Similarly, if $f : X \rightarrow Y$ is a morphism of relational structures and $x \in X_s$, then we often denote the image of $x$ under $f$ by $f(x)$ instead of $f_s(x)$.
If the signature $\mathfrak{S}$ is clear from context, we abbreviate $\mathrm{Rel}(\mathfrak{S})$ as $\mathrm{Rel}$.

There is an evident forgetful functor $\mathrm{Rel}(\mathfrak{S}) \rightarrow \mathrm{Set}^S$ to the $S$-ary product of the category of sets, which is given by discarding the relations.
When we mention the \emph{carrier sets} of a relational structure, we mean the result of applying this forgetful functor.

\begin{proposition}
  \label{prop:adjoints-signature-inclusion}
  Let $\mathfrak{S} = (S, R)$ and $\mathfrak{S}' = (S', R')$ be relational signatures such that $\mathfrak{S}'$ extends $\mathfrak{S}$, in the sense that $S \subseteq S', R \subseteq R'$ and $\mathfrak{S}$ and $\mathfrak{S}'$ assign the same arities to relation symbols $r \in R$.
  Then the evident forgetful functor $\mathrm{Rel}(\mathfrak{S}') \rightarrow \mathrm{Rel}(\mathfrak{S})$ has both a left adjoint and a right adjoint.
  Both adjoints are sections to the forgetful functor, that is, both composites
  \begin{equation}
    \begin{tikzcd}
      \mathrm{Rel}(\mathfrak{S}) \arrow[r, shift left] \arrow[r, shift right] & \mathrm{Rel}(\mathfrak{S}') \arrow[r] & \mathrm{Rel}(\mathfrak{S})
    \end{tikzcd}
  \end{equation}
  are identity functors.
\end{proposition}
\begin{proof}
  Let $X$ be a relational $\mathfrak{S}$-structure.
  Let $s \in S' \setminus S$ and let $r : s_1 \times \dots \times s_n$ be in $R' \setminus R$.
  The left adjoint extends $X$ to a relational $\mathfrak{S}'$-structure $Y$ by $Y_s = \emptyset$ and $r_{Y} = \emptyset$.
  The right adjoint extends $X$ to a relational $\mathfrak{S}'$ structure $Z$ such that $Z_s = \{ * \}$ is a singleton set and $r_Z = Z_{s_1} \times \dots \times Z_{s_n}$.
\end{proof}

\begin{proposition}
  \label{prop:rel-cocomplete}
  Let $\mathfrak{S} = (S, R)$ be a relational signature.
  Then $\mathrm{Rel}(\mathfrak{S})$ is complete and cocomplete, and the forgetful functor $\mathrm{Rel}(\mathfrak{S}) \rightarrow \mathrm{Set}^S$ preserves limits and colimits.
\end{proposition}
\begin{proof}
  Limit and colimit preservation follows from Proposition~\ref{prop:adjoints-signature-inclusion}, since the forgetful functor is induced by the extension of signatures $(S, \emptyset) \subseteq (S, R)$.
  Limits commute with other limits and in particular products.
  Thus, limits of relational structures can be constructed as limits of carriers endowed with the limits of relation sets.

  The construction of colimits is more involved because products do not generally commute with quotients.
  Let $D : I \rightarrow \mathrm{Rel}(\mathfrak{S})$ be a diagram of relational structures.
  We define the carrier sets of our candidate colimit structure $X$ by the colimit of carrier sets.
  That is,
  \begin{equation}
    X_s = \operatorname*{colim}_{i \in I} D(i)_s
  \end{equation}
  for all $s \in S$.
  We obtain evident maps $(p_i)_s : D(i)_s \rightarrow X_s$ for all objects $i$ in $I$ and $s \in S$.
  Let $r : s_1 \times \dots \times s_n$ be a relation symbol.
  Then we define $r_X$ as union over the images of the $r_{D(i)}$.
  Thus,
  \begin{equation}
    r_X = \bigcup_{i \in I} p_i(r_{D(i)})
  \end{equation}
  where $p_i(r_{D(i)}) = ((p_i)_{s_1} \times \dots \times (p_i)_{s_n})(r_{D(i)})$.
\end{proof}

\begin{definition}
  \label{def:finite-relational-structure}
  Let $\mathfrak{S} = (S, R)$ be a relational signature.
  A relational $\mathfrak{S}$-structure $X$ is \emph{finite} if
  \begin{equation}
    \label{eq:finite-relational-structure}
    \sum_{s \in S} |X_s| + \sum_{r \in R} |r_X| < \infty,
  \end{equation}
  that is, if all the $X_s$ and $r_X$ are finite and almost always empty.
\end{definition}

\begin{proposition}
  \label{prop:finite-equiv-finitely-presentable}
  Let $\mathfrak{S} = (S, R)$ be a relational signature.
  Then a relational $\mathfrak{S}$-structure is finite if and only if it is a finitely presentable object in $\mathrm{Rel}(\mathfrak{S})$.
\end{proposition}
\begin{proof}
  Let $X$ be a finite relational structure and let 
  \begin{equation}
    f : X \rightarrow Y = \operatorname{colim} D = \coprod_{i \in I} D(i) / \sim
  \end{equation}
  be a map to a filtered colimit.
  Then the image of each element $x \in X$ is represented by some element $y_x \in D(i_x)$.
  Since $X$ contains only finitely many elements and $D$ is directed, we may assume that $i_x = i_{x'} = i$ is constant over all $x, x' \in X$.
  For each tuple $t = (x_1, \dots, x_n) \in r_X$ for some $r \in R$ we have that $([y_{x_1}], \dots, [y_{x_n}]) \in r_Y$.
  Since there are only finitely many $t$, we may again increase $i$ so that $(y_{x_1}, \dots, y_{x_n}) \in r_{D(i)}$.
  Now $X \rightarrow Y$ factors via $D(i)$.

  Conversely, if a relational structure $X$ is not finite, then there exists a strictly increasing sequence of relational substructures
  \begin{equation}
    X_0 \subset X_1 \subset X_2 \subset \dots \subset X
  \end{equation}
  such that $\bigcup_{n \geq 0} X_n = X$.
  Then the canonical map $X = \bigcup_{n \geq 0} X_n \cong \operatorname{colim}_{n \geq 0} X_n$ does not factor via any $X_n$, so $X$ is not finitely presentable.
\end{proof}

\subsection{Syntax and Semantics}
\label{subsec:rhl-semantics}

Fix a relational signature $\mathfrak{S} = (S, R)$.
We assume a countable supply of variable symbols $v$, each annotated with a sort $s \in S$.
\begin{definition}
  \label{def:relational-syntax}
  An \emph{RHL atom} is an expression of one of the following forms:
  \begin{enumerate}
    \item
      \label{itm:relation-atom}
      A \emph{relation atom} $r(v_1, \dots, v_n)$, where $r : s_1 \times \dots \times s_n$ is a relation symbol and the $v_i$ are variables of sort $s_i$ for all $i \in \{1, \dots, n\}$.
    \item
      A \emph{sort quantification atom} $v \downarrow$ where $v$ is a variable.
    \item
      An \emph{equality atom} $u \equiv v$, where $u$ and $v$ are variables of the same sort.
  \end{enumerate}
  An \emph{RHL formula} is a finite conjunction $\phi_1 \land \dots \land \phi_n$ of RHL atoms $\phi_i$.
  An \emph{RHL sequent} is an implication $\mathcal{F} \implies \mathcal{G}$ of RHL formulas $\mathcal{F}, \mathcal{G}$.
  An \emph{RHL theory} is a set of RHL sequents.
\end{definition}
Observe that, since we assume that each variable has an intrinsically associated sort, annotating sort quantification atoms $v \downarrow$ with a sort $s$ would be redundant.
We reserve the ${\equiv}$ symbol for RHL syntax, while the ${=}$ symbol is used for meta-theoretical equality.
We always assume that meta-theoretical equality binds weaker than RHL connectives, so that $\mathcal{F} = \phi_1 \land \phi_2$ states that $\mathcal{F}$ is equal to the syntactic object $\phi_1 \land \phi_2$, and $\mathcal{S} = \mathcal{F} \implies \mathcal{G}$ states that $\mathcal{S}$ is equal to the sequent $\mathcal{F} \implies \mathcal{G}$.

\begin{definition}
  \label{def:relational-semantics}
  Let $X$ be a relational structure.
  An \emph{interpretation of a set of variables $V$} in $X$ is a map $I$ that assigns to each variable $v \in V$ of sort $s$ an element $I(v) \in X_s$.
  An \emph{interpretation of an RHL atom $\phi$} in $X$ is an interpretation $I$ of the variables occurring in $\phi$ such that one of the following conditions holds:
  \begin{enumerate}
    \item
      $\phi = r(v_1, \dots, v_n)$ is a relation atom and $(I(v_1), \dots, I(v_n)) \in r_X$.
    \item
      $\phi = v \downarrow$ is a sort quantification atom, without further assumptions.
    \item
      $\phi = u \equiv v$ is an equality atom and $I(u) = I(v)$.
  \end{enumerate}
  An \emph{interpretation of an RHL formula $\mathcal{F} = \phi_1 \land \dots \land \phi_n$} in $X$ is an interpretation of the variables occurring in $\mathcal{F}$ that restricts to an interpretation of $\phi_i$ for each $i \in \{1, \dots, n\}$.

  A relational structure $X$ \emph{satisfies} an RHL sequent $\mathcal{F} \implies \mathcal{G}$ if each interpretation of $\mathcal{F}$ in $X$ can be extended to an interpretation of $\mathcal{F} \land \mathcal{G}$ in $X$.
  A \emph{model} of a theory $T$ is a relational structure that satisfies all sequents in $T$.
  The \emph{category of models} $\mathrm{Mod}(T)$ is the full subcategory of relational structures given by the models of $T$.
\end{definition}

\begin{definition}
  We associate to each RHL atom $\phi$ a \emph{classifying relational structure} $[\phi]$ and a \emph{generic interpretation} $I_\phi$ of $\phi$ in $[\phi]$ as follows:
  \begin{enumerate}
    \item
      If $\phi = r(v_1, \dots, v_n)$ where $r : s_1 \times \dots \times s_n$, then the carriers of $[\phi]$ are given by distinct elements $I_\phi(v_i) \in [\phi]_{s_i}$ and a single tuple $(I_\phi(v_1), \dots, I_\phi(v_n)) \in r_{[\phi]}$.
      The relations $r'_{[\phi]}$ for $r \neq r'$ are empty.
    \item
      If $\phi = v \downarrow$, where $v$ has sort $s$, then $[\phi]_s$ contains a single element $I_\phi(v)$.
      All other carrier sets and all relations are empty.
    \item
      If $\phi = v_1 \equiv v_2$, where $v_1$ and $v_2$ have sort $s$, then $[\phi]_s$ contains a single element $I_\phi(v_1) = I_\phi(v_2)$.
      All other carrier sets and all relations are empty.
  \end{enumerate}

  Let $\mathcal{F} = \phi_1 \land \dots \land \phi_n$ be an RHL formula.
  The \emph{classifying relational structure} $[\mathcal{F}]$ of $\mathcal{F}$ is the quotient
  \begin{equation}
    ([\phi_1] \amalg \dots \amalg [\phi_n]) / \sim
  \end{equation}
  where $\sim$ is the relation given by
  \begin{equation}
    I_{\phi_i}(v) \sim I_{\phi_j}(v)
  \end{equation}
  for all $i, j \in \{1, \dots, n\}$ and variables $v$ occurring in both $\phi_i$ and $\phi_j$, and the generic interpretation $I_\mathcal{F}$ is the amalgamation of the interpretations $I_{\phi_i}$.
\end{definition}

\begin{proposition}
  \label{prop:relational-formula-structure-universal}
  Let $\mathcal{F}$ be an RHL formula and let $X$ be a relational structure.
  Then there is a bijection between interpretations of $\mathcal{F}$ in $X$ and maps $[\mathcal{F}] \rightarrow X$.
\end{proposition}
\begin{proof}
  If $f : [\mathcal{F}] \rightarrow X$, then $f \circ I_\mathcal{F}$ is an interpretation of $\mathcal{F}$ in $X$.
  Conversely, let $\mathcal{F} = \phi_1 \land \dots \land \phi_n$ for RHL atoms $\phi_i$.
  Then every interpretation $I$ of $\mathcal{F}$ restricts to an interpretation of $\phi_i$ for each $i$.
  The carrier sets of $[\phi_i]$ are defined using the variables of $\phi_i$, which defines an evident map $[\phi_i] \rightarrow X$.
  Since the restrictions of $I$ to the variables in each $\phi_i$ agree on variables that occur simultaneously in two atoms, the individual maps $[\phi_i] \rightarrow X$ glue to a map $[\mathcal{F}] \rightarrow X$.
\end{proof}

\begin{definition}
  Let $\mathcal{S} = \mathcal{F} \implies \mathcal{G}$ be an RHL sequent.
  The \emph{classifying morphism} of $\mathcal{S}$ is the map $[\mathcal{S}] : [\mathcal{F}] \rightarrow [\mathcal{F} \land \mathcal{G}]$ that is induced by the canonical interpretation of $\mathcal{F}$ in $[\mathcal{F} \land \mathcal{G}]$.
\end{definition}

\begin{proposition}
  \label{prop:relational-sequent-morphism-lifts}
  Let $\mathcal{S}$ be an RHL sequent and let $X$ be a relational structure.
  Then $X$ satisfies $\mathcal{S}$ if and only if $X$ is injective to $[\mathcal{S}]$.
\end{proposition}
\begin{proof}
  Let $\mathcal{S} = \mathcal{F} \implies \mathcal{G}$.
  By Proposition~\ref{prop:relational-formula-structure-universal}, interpretations $I$ of the premise $\mathcal{F}$ correspond to maps $\langle I \rangle : [\mathcal{F}] \rightarrow X$, and interpretations $J$ of $\mathcal{F} \land \mathcal{G}$ correspond to maps $\langle J \rangle : [\mathcal{F} \land \mathcal{G}] \rightarrow X$.
  The map $[\mathcal{S}] : [\mathcal{F}] \rightarrow [\mathcal{F} \land \mathcal{G}]$ is given by restriction of the generic interpretation of $\mathcal{F} \land \mathcal{G}$ in $[\mathcal{F} \land \mathcal{G}]$ to the variables occurring in $\mathcal{F}$.
  It follows that
  \begin{equation}
    \begin{tikzcd}
      \left[\mathcal{F}\right] \arrow[r, "\langle I \rangle"] \arrow[d, "{[\mathcal{S}]}"'] & X \\
      \left[\mathcal{F} \land \mathcal{G}\right] \arrow[ur, "\langle J \rangle"']
    \end{tikzcd}
  \end{equation}
  commutes if and only if $I$ is a restriction of $J$.
\end{proof}

\begin{proposition}
  \label{prop:composition-rhl-sequent-morphisms}
  Let $\mathcal{F}, \mathcal{G}$ and $\mathcal{H}$ be RHL formulas.
  Then
  \begin{equation}
    [\mathcal{F} \land \mathcal{G} \implies \mathcal{H}] \circ [\mathcal{F} \implies \mathcal{G}] \cong [\mathcal{F} \implies \mathcal{G} \land \mathcal{H}].
  \end{equation}
\end{proposition}
\begin{proof}
  By the universal property of classifying relational structures.
\end{proof}

\begin{definition}
  An RHL theory $T$ is \emph{strong} if $[T] = \{ [\mathcal{S}] \mid \mathcal{S} \in T \}$ is strong.
\end{definition}

\begin{proposition}
  \label{prop:free-models}
  Let $T$ be an RHL theory.
  Then $\mathrm{Mod}(T) \subseteq \mathrm{Rel}(\mathfrak{S})$ is a weakly reflective category.
  If $T$ is strong, then $\mathrm{Mod}(T) \subseteq \mathrm{Rel}(\mathfrak{S})$ is a reflective subcategory.
\end{proposition}
\begin{proof}
  By application of the small object argument (Propositions \ref{prop:small-object-argument-prop} and \ref{prop:small-object-argument-exist}) to $M = \{ [\mathcal{S}] \mid \mathcal{S} \in T \}$.
\end{proof}

\begin{remark}
  Let us reflect on the similarities between the proof of Proposition~\ref{prop:free-models} and Datalog evaluation.
  Note that Datalog is a strict subset of RHL, so we have to specialize Proposition~\ref{prop:free-models} to Datalog theories $T$ in order to compare.
  Thus, we assume that the conclusions of sequents in $T$ only contain atoms $r(v_1, \dots, v_n)$ for variables that occur in the premise (see Section~\ref{subsec:datalog-vs-rhl} for detailed discussion of the fragment of RHL corresponding to Datalog).

  Unfolding the small object argument, we see that the reflection of a relational structure $X$ into the category of models is given by the colimit of a chain
  \begin{equation}
    \label{eq:datalog-small-object-argument}
    \begin{tikzcd}
      X = X_0 \arrow[r] & X_1 \arrow[r] & X_2 \arrow[r] & \dots
    \end{tikzcd}
  \end{equation}
  of relational structures.
  The relational structure $X$ corresponds to the set of input facts of the Datalog program, and each $X_i$ represents the total set of derived facts after the $i$th iteration of Datalog evaluation.
  Because $T$ contains Datalog sequents only, the transition maps $X_i \rightarrow X_{i + 1}$ are bijective on carriers.
  The data of the sequence \ref{eq:datalog-small-object-argument} is thus equivalent to a sequence of inclusions
  \begin{equation}
    r_X = r_{X_0} \subseteq r_{X_1} \subseteq \dots
  \end{equation}
  on the carrier of $X$ for all relation symbols $r$, mirroring the monotonically growing relations during Datalog evaluation.

  Unfolding our existence proof of the small object argument (Proposition~\ref{prop:small-object-argument-exist}) and the universal property of classifying structures (Proposition~\ref{prop:relational-formula-structure-universal}), we see that $X_{i + 1}$ is obtained from $X_i$ via the following pushout square:
  \begin{equation}
    \label{eq:small-object-argument-step-datalog}
    \begin{tikzcd}
      \coprod_{(\mathcal{S}, I) \in K} [\mathcal{F}_\mathcal{S}] \arrow[r] \arrow[d] \arrow[dr, phantom, very near end, "\ulcorner"] & \coprod_{(\mathcal{S}, I) \in K} [\mathcal{G}_\mathcal{S}] \arrow[d] \\
      X_i \arrow[r] & X_{i + 1}
    \end{tikzcd}
  \end{equation}
  Here $K$ is the set of pairs of sequents $\mathcal{S} = \mathcal{F}_\mathcal{S} \implies \mathcal{G}_\mathcal{S}$ and interpretations $I : [\mathcal{F}_\mathcal{S}] \rightarrow X_i$.
  Thus the left vertical map corresponds to the set of matches of premises among the facts established after the $i$th iteration of Datalog evaluation.
  Defining $X_{i + 1}$ using the pushout square above has the effect of adjoining matches of the conclusion for each match of the premise.
  Since $T$ is a Datalog theory, the conclusions are relation atoms, hence $X_{i + 1}$ is obtained from $X_i$ by adjoining tuples to relations.
\end{remark}

\begin{remark}
  Semi-naive evaluation is an optimized version of Datalog evaluation, where we consider only matches of premises at the $i$th stage that have not been present already in the $(i - 1)$th stage.
  This does not change the result of Datalog evaluation since conclusions of matches that have been found in a previous iteration have already been adjoined.
  In terms of the small object argument, this optimization can be understood as a more economic choice of the set $K$:
  In diagram \ref{eq:small-object-argument-step-datalog}, we can replace $K$ by the set of interpretations $I : [\mathcal{F}_\mathcal{S}] \rightarrow X_i$ that do not factor via $X_{i - 1}$.
\end{remark}

\begin{remark}
  Still, there are properties of Datalog evaluation that Proposition~\ref{prop:free-models} does not entirely capture.
  First, the result of Datalog evaluation is determined uniquely via fixed point semantics, whereas Proposition~\ref{prop:free-models} guarantees uniqueness (up to ismorphism) only in the case of strong theories.
  Since classifying morphisms of Datalog sequents are epic, all Datalog theories are strong (Proposition~\ref{prop:epi-strong}).
  Thus, Proposition~\ref{prop:free-models} does indeed determine the result of Datalog computation uniquely.
  However, not all strong theories are Datalog theories or even contain epimorphisms only.
  For example, Proposition~\ref{prop:strongbization} allows extending every RHL theory to a strong theory.
  A syntactic characterization of strong RHL theories is the main purpose of Section~\ref{sec:phl}, where we discuss partial Horn logic.

  A second feature of Datalog evaluation we have not discussed is that it always terminates:
  Since Datalog evaluation monotonically increases the size of relations on a fixed carrier, it reaches a fixed point after a finite number of iterations.
  This is not generally true for RHL theories, since the carrier sets change during evaluation.
  We can, however, prove termination for \emph{surjective} theories, which subsume and generalize Datalog theories (Corollary~\ref{cor:surjective-rhl-theories-terminate}).
  In general, it is undecidable whether evaluating a given RHL theory terminates on some input.
\end{remark}

\subsection{Completeness Results}
\label{subsec:rhl-completeness}

In this section, we study the descriptive strength of RHL.
We show that every morphism of finite relational structures can be described as classifying morphism of an RHL sequent (Proposition~\ref{prop:image-of-classifying-rhl-sequents}).
We then define identify subsets of RHL corresponding to injections and surjections of finite relational structures.
Finally, we show that the sequence resulting from application of the small object argument applied to surjective RHL sequents reaches a fixed point after a finite number of steps (Corollary~\ref{cor:surjective-rhl-theories-terminate}).
This generalizes the fact that evaluation of Datalog programs always terminates.

\begin{proposition}
  \label{prop:image-of-classifying-rhl-sequents}
  Let $\mathcal{S}$ be an RHL sequent.
  Then the classifying morphism $[\mathcal{S}]$ is a morphism of finite relational structures.
  Conversely, for every morphism $f : X \rightarrow Y$ of finite relational structures, there exists an RHL sequent $\mathcal{S} = \mathcal{F} \implies \mathcal{G}$ such that $f$ and $[\mathcal{S}]$ are isomorphic, in the sense that there exists a commutative square of the form
  \begin{equation}
    \begin{tikzcd}
      X \arrow[r, "f"] \arrow[d, "\cong"'] & Y \arrow[d, "\cong"] \\
      \left[\mathcal{F}\right] \arrow[r, "{[\mathcal{S}]}"] & \left[\mathcal{F} \land \mathcal{G}\right].
    \end{tikzcd}
  \end{equation}
\end{proposition}
\begin{proof}
  Since RHL formulas are finite, it follows from construction that classifying relational structures are finite.
  Conversely, let $f : X \rightarrow Y$ be a morphism of finite relational structures.
  Choose distinct variables $v_x$ of sort $s$ for every sort $s$ and element $x \in X_s$.
  The RHL formulas
  \begin{mathpar}
    \mathcal{F}_{\mathrm{car}} = \bigwedge_{x \in X} v_x \downarrow
    \and
    \mathcal{F}_{\mathrm{rel}} = \bigwedge_{\substack{r \in R\\(x_1, \dots, x_n) \in r_X}} r(v_{x_1}, \dots, v_{x_n})
  \end{mathpar}
  are finite (and hence well-defined) because $X$ is finite.
  The formula $\mathcal{F}_{\mathrm{car}}$ encodes the carrier sets of $X$, and $\mathcal{F}_{\mathrm{rel}}$ encodes the relations.
  Thus if $\mathcal{F} = \mathcal{F}_{\mathrm{car}} \land \mathcal{F}_{\mathrm{rel}}$, then $X \cong [\mathcal{F}]$.

  Let $u_y$ be a fresh variable for each sort $s$ and element $y \in Y_s \setminus \operatorname{Im} f_s$ that is not in the image of $f$.
  For arbitrary elements $y \in Y$, we let $w_y = v_x$ for some fixed choice of $x \in X$ such that $f(x) = y$ if such an element $x$ exists, and $w_y = u_y$ otherwise.

  Set
  \begin{mathpar}
    \mathcal{G}_\mathrm{car} = \bigwedge_{y \in Y \setminus \operatorname{Im} f} u_y \downarrow
    \and
    \mathcal{G}_\mathrm{eq} = \bigwedge_{\substack{x,y \in X \\ f(x) = f(y)}} v_{x} \equiv v_{y}
    \\
    \mathcal{G}_{\mathrm{rel}} = \bigwedge_{\substack{r \in R\\(y_1, \dots, y_n) \in r_Y}} r(w_{y_1}, \dots, w_{y_n})
  \end{mathpar}
  and let $\mathcal{G} = \mathcal{G}_{\mathrm{car}} \land \mathcal{G}_{\mathrm{rel}} \land \mathcal{G}_{\mathrm{eq}}$.
  Then $f$ has the universal property of the classifying morphism of $\mathcal{S} = \mathcal{F} \implies \mathcal{G}$, hence $f \cong [\mathcal{S}]$.
\end{proof}

\begin{definition}
  Let $f : X \rightarrow Y$ be a map of relational structures.
  \begin{enumerate}
    \item
      $f$ is \emph{injective} if $f_s : X_s \rightarrow Y_s$ is an injective function for all sorts $s \in S$.
    \item
      $f$ is \emph{surjective} if $f_s : X_s \rightarrow Y_s$ is a surjective function for all sorts $s \in S$.
  \end{enumerate}
\end{definition}

\begin{proposition}
  \label{prop:monic-epic-rel-morphisms}
  Let $f  : X \rightarrow Y$ be a morphism of relational structures.
  \begin{enumerate}
    \item
      \label{itm:injective-iff-mono}
      $f$ is injective if and only if $f$ is a monomorphism in $\mathrm{Rel}$.
    \item
      \label{itm:surjective-iff-epi}
      $f$ is surjective if and only if $f$ is an epimorphism in $\mathrm{Rel}$.
  \end{enumerate}
\end{proposition}
\begin{proof}
  \ref{itm:injective-iff-mono}.
  In general, a morphism $f : X \rightarrow Y$ in a complete category $\mathcal{C}$ is a monomorphism if and only if
  \begin{equation}
    \begin{tikzcd}
      X \arrow[r, "="] \arrow[d, "="'] & X \arrow[d, "f"] \\
      X \arrow[r, "f"] & Y
    \end{tikzcd}
  \end{equation}
  is a pullback square.
  Because the carrier functor, i.e. forgetful functor from relational structures to $S$-indexed families of sets preserves limits, it follows that it preserves monomorphisms.
  Thus, every monomorphism of relational structures is injective.
  The same forgetful functor is faithful, hence reflects monomorphisms, so every injective morphism of relational structures is a monomorphism.

  \ref{itm:surjective-iff-epi}.
  Analogously to \ref{itm:injective-iff-mono}, since the carrier functor also preserves colimits.
\end{proof}

\begin{definition}
  \label{def:monic-epic-rhl-sequents}
  Let $\mathcal{S} = \mathcal{F} \implies \mathcal{G}$ be an RHL sequent.
  \begin{enumerate}
    \item
      $\mathcal{S}$ is \emph{injective} if the conclusion $\mathcal{G}$ does not contain an equality atom.
    \item
      $\mathcal{S}$ is \emph{surjective} if every variable in the conclusion $\mathcal{G}$ also occurs in the premise $\mathcal{F}$.
  \end{enumerate}
\end{definition}

\begin{proposition}
  \label{prop:classifying-monic-epic}
  The classifying morphisms of RHL sequents with the properties of Definition~\ref{def:monic-epic-rhl-sequents} can be characterized up to isomorphism as follows:
  \begin{enumerate}
    \item
      \label{itm:classifying-injective}
      The classifying morphisms of injective RHL sequents are precisely the injections of finite relational structures.
    \item
      \label{itm:classifying-surjective}
      The classifying morphisms of surjective RHL sequents are precisely the surjections of finite relational structures.
  \end{enumerate}
\end{proposition}
\begin{proof}
  The verification that the classifying morphisms of the sequents in question have the desired properties can be reduced to sequents of the form $\mathcal{F} \implies \phi$ where $\phi$ is an RHL atom by Proposition~\ref{prop:composition-rhl-sequent-morphisms}, and then follows by case distinction on $\phi$.

  Conversely, the construction of sequents $\mathcal{S}$ with the respective property given a morphism $f : X \rightarrow Y$ such that $f \cong [\mathcal{S}]$ is analogous to the proof of Proposition~\ref{prop:image-of-classifying-rhl-sequents}.
  In both cases, the atoms in the conclusion $\mathcal{G}$ that violate the condition on the sequent are redundant because of the assumed property of $f$:
  If $f$ is injective, then $\mathcal{G}_\mathrm{eq}$ is a conjunction of equality atoms of the form $v \equiv v$ and hence can be omitted.
  If $f$ is surjective, then $Y \setminus \operatorname{Im} f$ is empty, so the case $w_y = u_y$ for some $y$ that is not in the image of $f$ does not occur.
\end{proof}

\begin{proposition}
  \label{prop:surjective-small-object-argument}
  Let $M$ be a finite set of epimorphisms of finite relational structures.
  Let
  \begin{equation}
    \begin{tikzcd}
      X_0 \arrow[r, "x_0"] & X_1 \arrow[r, "x_1"] & \dots
    \end{tikzcd}
  \end{equation}
  be any sequence of maps of relational structures satisfying the conditions of Proposition~\ref{prop:small-object-argument-prop} such that furthermore $X_0$ is finite.
  Then the sequence is eventually stationary, in the sense that $x_n$ is an isomorphism for all sufficiently large $n$.
\end{proposition}
\begin{proof}
  Since all maps in $M$ are surjective and colimits of relational structures commute with colimits on carrier sets, it follows that all maps in $\mathrm{Cell}(M)$ are surjective.
  Thus the cardinality of the carriers of the $X_n$ decreases monotonically with $n$.
  Since $X$ is finite, the carriers $X_s$ are empty for almost all sorts $s$.
  Eventually, the sum of the cardinalities of the carriers of $X_n$ must thus become stable, say after $n_0 \in \mathbb{N}$.
  Without loss of generality, we may assume that $x_n$ is the identity map on carriers for $n \geq n_0$.
  Let $r \in R$.
  Then for all $n \geq n_0$, we have that
  \begin{equation}
    r_{X_n} \subseteq r_{X_{n + 1}} \subseteq (X_{n_0})_{s_1} \times \dots \times (X_{n_0})_{s_n}
  \end{equation}
  and the latter is a finite set.
  Thus, eventually $r_{X_n} = r_{X_{n + 1}}$ is stationary.
  Even when $R$ is infinite, we have $r_X = r_{X_n}$ for all $n$ and almost all $r$, since only finitely many relations are non-empty in any of the involved relational structures ($X_{n_0}$ or a domain or codomain of a map in $M$).
  For sufficiently large $n_1 \in \mathbb{N}$ and all $n \geq n_1$ we thus have $r_{X_n} = r_{X_{n + 1}}$ for all $r$ and hence $X_n = X_{n + 1}$.
\end{proof}

\begin{corollary}
  \label{cor:surjective-rhl-theories-terminate}
  Let $T$ be an RHL theory containing only surjective sequents.
  Then the reflection of a finite relational structure into $\mathrm{Mod}(T)$ is a finite relational structure.
  \qed
\end{corollary}

\subsection{Datalog and Relational Horn Logic}
\label{subsec:datalog-vs-rhl}

In this section, we study the subset of RHL that corresponds to Datalog.
We show that RHL can be reduced to \emph{Datalog with choice} by way of the \emph{setoid transformation}.

\begin{definition}
  Let $\mathcal{S}$ be an RHL sequent.
  \begin{enumerate}
    \item
      \label{itm:plain-datalog}
      $\mathcal{S}$ is a \emph{Datalog sequent} if all atoms in $\mathcal{S}$ are of the form $r(v_1, \dots, v_n)$ and all variables in the conclusion of $\mathcal{S}$ also occur in the premise.
    \item
      $\mathcal{S}$ is a \emph{Datalog sequent with sort quantification} if all atoms in $\mathcal{S}$ are of the form $r(v_1, \dots, v_n) \downarrow$ or $v \downarrow$, and all variables in the conclusion of $\mathcal{S}$ also occur in the premise.
    \item
      \label{itm:datalog-with-choice}
      $\mathcal{S}$ is a \emph{Datalog sequent with choice} if all atoms in $\mathcal{S}$ are of the form $r(v_1, \dots, v_n) \downarrow$ or $v \downarrow$.
  \end{enumerate}
\end{definition}

Note that in standard Datalog, usually only sequents with a single atom as conclusion are allowed.
However, our generalized Datalog sequents \ref{itm:plain-datalog} have the same descriptive power as standard Datalog, since a single sequent with $n$ conclusions can equivalently be replaced by $n$ sequents with single conclusions.
The name \emph{Datalog with choice} in \ref{itm:datalog-with-choice} alludes to the choice construct in Souffle \citep{souffle-choice} with similar semantics.

\begin{definition}
  An element $x \in X$ in a relational structure is \emph{unbound} if it does not appear in any tuple $t \in r_X$ for all $r \in R$.
\end{definition}

\begin{proposition}
  The classifying morphisms of Datalog sequents can be characterized up to isomorphism as follows:
  \begin{enumerate}
    \item
      \label{itm:semantics-datalog}
      The classifying morphisms of Datalog sequents are precisely the injective surjective morphisms of finite relational structures that do not contain unbound variables.
    \item
      \label{itm:semantics-datalog-with-sort-quantification}
      The classifying morphisms of Datalog sequents with sort quantification are precisely the injective surjective morphisms of finite relational structures.
    \item
      \label{itm:semantics-datalog-with-choice}
      The classifying morphisms of Datalog sequents with choice are precisely the injective surjective morphisms of finite relational structures.
  \end{enumerate}
\end{proposition}
\begin{proof}
  Analogously to the proofs of Propositions \ref{prop:image-of-classifying-rhl-sequents} and \ref{prop:classifying-monic-epic}.
\end{proof}

\begin{definition}
  A \emph{setoid} consists of a set $X$ and an equivalence relation $\sim_X$ on $X$.
  A morphism $f : X \rightarrow Y$ is a map of underlying sets that respects the equivalence relations.
  Two morphisms $f, g : X \rightarrow Y$ of setoids are equal if $f(x) \sim_Y g(x)$ for all $x \in X$.
  The category of setoids is denoted by $\mathrm{Setoid}$.
\end{definition}

\begin{proposition}
  \label{prop:setoid-quotient-equivalence}
  The categories $\mathrm{Setoid}$ and $\mathrm{Set}$ are equivalent.
  An equivalence is given by the quotient functor $\mathrm{Setoid} \rightarrow \mathrm{Set}, (X, \sim_X) \mapsto X / \sim_X$ and the diagonal functor $\mathrm{Set} \rightarrow \mathrm{Setoid}, X \mapsto (X, \{(x, x) \mid x \in X\})$.
  \qed
\end{proposition}

\begin{definition}
  \label{def:setoid-transformation}
  The \emph{setoid transformation} of an RHL theory $T$ defined on a signature $\mathfrak{S} = (S, R)$ is a Datalog with choice theory $T'$ defined on a relational signature $\mathfrak{S}'$ as follows.
  The signature $\mathfrak{S}'$ extends $\mathfrak{S}$ by a relation symbol $\mathrm{Eq}_s : s \times s$ for each sort $s \in S$.
  The sequents of $T'$ are given as follows:
  \begin{enumerate}
    \item
      \label{itm:setoid-theory-equivalence}
      For each sort $s$, sequents asserting that $\mathrm{Eq}_s$ is an equivalence relation:
      \begin{mathpar}
        x! \implies \mathrm{Eq}_s(x)
        \and
        \mathrm{Eq}_s(x, y) \implies \mathrm{Eq}_s(y, x)
        \\
        \mathrm{Eq}_s(x, y) \land \mathrm{Eq}_s(y, z) \implies \mathrm{Eq}_s(x, z)
      \end{mathpar}
    \item
      \label{itm:setoid-theory-congruence}
      For each relation $r : s_1 \times \dots \times s_n$, a sequent asserting that the equivalence relations $\mathrm{Eq}_s$ behave as congruences with respect to $r$:
      \begin{equation}
        r(v_1, \dots, v_n) \land \mathrm{Eq}_{s_1}(v_1, u_1) \land \dots \land \mathrm{Eq}_{s_n}(v_n, u_n) \implies r(u_1, \dots, u_n)
      \end{equation}
    \item
      For each sequent $\mathcal{S}$ in $T$, the sequent which is obtained from $\mathcal{S}$ by replacing each equality atom $u \equiv v$ with the atom $\mathrm{Eq}_s(u, v)$, where $s$ is the sort of $u$ and $v$.
  \end{enumerate}
  The \emph{category of setoid models} $\mathrm{Mod}_\mathrm{Setoid}(\mathfrak{S}, T)$ is given by the models of $(\mathfrak{S}', T')$, where we consider morphisms of setoid models $f, g : X \rightarrow Y$ as equal if $f_s, g_s : (X_s, \mathrm{Eq}_s) \rightarrow (Y_s, \mathrm{Eq}_s)$ are equal as setoid morphisms for all sorts $s$.
\end{definition}

\begin{proposition}
  \label{prop:setoid-model-equivalence}
  Let $(\mathfrak{S}, T)$ be an RHL theory.
  Then $\mathrm{Mod}_\mathrm{Setoid}(\mathfrak{S}, T)$ and $\mathrm{Mod}(\mathfrak{S}, T)$ are equivalent categories.

  An equivalence is given as follows.
  The functor $F : \mathrm{Mod}(\mathfrak{S}, T) \rightarrow \mathrm{Mod}_\mathrm{Setoid}(\mathfrak{S}, T)$ extends a relational $\mathfrak{S}$-structure $X$ to a relational $\mathfrak{S}'$-structure $F(X)$ on the same carrier by $(\mathrm{Eq}_s)_{F(X)} = \{(x, x) \mid x \in X\}$ for all sorts $s \in S$.
  The functor $G : \mathrm{Mod}_\mathrm{Setoid}(\mathfrak{S}, T)$ assigns to a setoid model $Y$ the relational $\mathfrak{S}$-structure with carriers $X_s = Y_s / \mathrm{Eq}_s$ and relations $r_X = \{ ([y_1], \dots, [y_n]) \mid (y_1, \dots, y_n) \in r_Y \}$.
\end{proposition}
\begin{proof}
  We must first verify that $F$ and $G$ are well-defined, i.e., that the relational structures in their images are indeed models of the respective theories.
  This is clear for $F$.

  Let $X = G(Y)$ for $Y \in \mathrm{Mod}_\mathrm{Setoid}(\mathfrak{S}, T)$.
  Let $\mathcal{F}$ be a formula for the signature $\mathfrak{S}$ and let $I$ be an interpretation of $\mathcal{F}$ in $X$.
  Since the carriers of $X$ are defined as quotients of the carriers of $Y$, every interpretation $I$ of $\mathcal{F}$ in $X$ lifts to an interpretation $I'$ of the same set of variables in $Y$, so that we have $I(v) = [I'(v)]$ for all variables $v$.
  Note that $I'$ is an interpretation of the set of variables of $\mathcal{F}$, but not always of the formula $\mathcal{F}$.
  Let $\mathcal{F}'$ be the formula obtained from $\mathcal{F}$ by replacing every equality atom $u \equiv v$ by the atom $\mathrm{Eq}_s(u, v)$, where $s$ is the sort of $u$ and $v$.
  We claim that $I'$ is an interpretation of $\mathcal{F}'$ in $Y$.
  To show this, it suffices to consider the case where $\mathcal{F}$ is an atom.
  \begin{itemize}
    \item
      If $\mathcal{F} = u \equiv v$, then $I(u) = I(v)$, so $[I'(u)] = I(u) = I(v) = [I'(v)]$.
      Thus $I'(u)$ and $I'(v)$ are in the same equivalence class, that is, $(I'(u), I'(v)) \in (\mathrm{Eq}_s)_Y$.
    \item
      If $\mathcal{F} = r(v_1, \dots, v_n)$ for some relation symbol $r$, then $(I(v_1), \dots, I(v_n)) \in r_Y$.
      By definition of $r_Y$, there exist $y_1, \dots, y_n \in Y$ such that $[y_i] = I(v_i)$ and $(y_1, \dots, y_n) \in r_X$.
      Thus $[I'(v_i)] = [y_i]$, so we have $(I'(v_i), y_i) \in (\mathrm{Eq}_{s_i})_Y$ for all $i$.
      Since $Y$ satisfies the congruence sequents \ref{itm:setoid-theory-congruence}, $r_Y$ is closed under equivalence in each argument, hence $(I'(v_1), \dots, I'(v_n)) \in r_Y$.
    \item
      The case $\mathcal{F} = v \downarrow$ is trivial.
  \end{itemize}
  Conversely, every interpretation $I'$ of $\mathcal{F}'$ in $X$ descends to an interpretation of $I$ of $\mathcal{F}$ in $Y$ by setting $I'(v) = [I(v)]$.
  
  Now, let $\mathcal{F} \implies \mathcal{G}$ be a sequent in $T$, and let $I$ be an interpretation of $\mathcal{F}$ in $X$.
  We have just shown that  $I'$ lifts to an interpretation of $\mathcal{F}'$ in $Y$.
  Because $Y$ satisfies $\mathcal{F}' \implies \mathcal{G}'$, we can extend $I'$ to an interpretation $J'$ of $\mathcal{G}'$ in $Y$, and then $J'$ descends to an interpretation of $\mathcal{G}$ in $X$ that extends $I$.
  Thus $G$ is well-defined.

  The composition $G \circ F$ is equivalent to the identity functor since a quotient by the diagonal does not change the original set.
  As for $F \circ G$, note that there is a canonical map $f : Y \rightarrow F(G(Y))$ for all setoid models $Y$.
  The restriction $f_s$ of $f$ to a setoid carrier $(X_s, \mathrm{Eq}_s)$ is an isomorphism of setoids for all sorts $s$, with inverses $g_s$ given by a choice of representative in each equivalence class.
  Since the relations of $X$ are closed under equivalence in each argument, it follows that $g$ is a morphism of relational structures.
  Thus $f$ and $g$ are isomorphisms.
\end{proof}

\begin{corollary}
  \label{cor:setoid-model-free}
  Let $(\mathfrak{S}, T)$ be an RHL theory with setoid transformation $(\mathfrak{S}', T')$.
  Then the reflection $\mathrm{Rel}(\mathfrak{S}) \rightarrow \mathrm{Mod}(\mathfrak{S}, T)$ can be computed as composite
  \begin{equation}
    \mathrm{Rel}(\mathfrak{S}) \xrightarrow{F_1} \mathrm{Rel}(\mathfrak{S}') \xrightarrow{F_2} \mathrm{Mod}(\mathfrak{S}', T') \xrightarrow{F_3} \mathrm{Mod}(\mathfrak{S}, T)
  \end{equation}
  where
  \begin{itemize}
    \item
      $F_1$ is the functors that extends relational $\mathfrak{S}$-structures $X$ to relational $\mathfrak{S}'$-structures with empty relations $\mathrm{Eq}_s$,
    \item
      $F_2$ is the free $T'$-model functor, and
    \item
      $F_3$ is one half of the equivalence constructed in Proposition~\ref{prop:setoid-model-equivalence}.
  \end{itemize}
\end{corollary}
\begin{proof}
  $F_1$ and $F_2$ are left adjoints and $F_3$ is an equivalence.
  The composite of the respective right adjoints is the inclusion $\mathrm{Mod}(\mathfrak{S}, T) \subseteq \mathrm{Rel}(\mathfrak{S})$, so the composite of the $F_i$ is the reflection into $\mathrm{Mod}(\mathfrak{S}, T)$.
\end{proof}

\begin{remark}
  Proposition~\ref{prop:setoid-model-equivalence} shows that RHL can be reduced to Datalog with minor extensions.
  In practice, however, using the resulting Datalog programs to compute free models is often unfeasible even for small inputs.

  One issue is that storing the equivalence relations $\mathrm{Eq}_s$ naively requires quadratic memory with respect to the size of equivalence classes.
  This problem can be largely addressed by using a union-find data structure, which only requires linear memory.
  Union-find data structures are available in the Souffle Datalog engine \citep{souffle-union-find}.

  A more significant issue are the congruence axioms \ref{itm:setoid-theory-congruence}, which result in an exponential increase in memory requirements with respect to the arity of relations.
  Every equality inferred during evaluation can significantly increase the total size of the relational structure in the next stage.
  Semantically, however, every inferred equality should in fact \emph{reduce} the size of the relational structure at the next stage.
  The Eqlog engine, which evaluates RHL theories directly, takes advantage of this observation by maintaining a union-find data structure on each sort to keep track of a canonical representative for each equivalence class.
  The relations then contain entries only for these canonical representatives.
  An inferred equality results in a merge of two equivalence classes, with one of the canonical representatives ceasing to be representative.
  Eqlog then \emph{canonicalizes} all tuples by replacing each occurrence of the old representative with the new representative.
  Since relations are stored without duplicates, this often results in a decrease in the size of the relations, so that later stages can be computed faster.
\end{remark}

\begin{remark}
  The following alternative \emph{sparse} setoid transformation $(\mathfrak{S}', T'')$ of an RHL theory $(\mathfrak{S}, T)$ can result in a more efficient Datalog program.
  The signature $\mathfrak{S}'$ of the sparse setoid transformation is the same as in the standard setoid transformation (Definition~\ref{def:setoid-transformation}), so $\mathfrak{S}'$ contains additional equivalence relations $\mathrm{Eq}_s$ for all sorts $s$.
  As before, $T''$ contains the equivalence relation axioms \ref{itm:setoid-theory-equivalence}.
  However, the congruence axioms \ref{itm:setoid-theory-congruence} are omitted.

  Instead, we modify the premise of each sequent $\mathcal{F} \implies \mathcal{G}$ in $T$ so that the different occurrences of a variable can be interpreted by distinct but equivalent elements.
  As before, we replace equality atoms $u \equiv v$ by atoms $\mathrm{Eq}_s(u, v)$ in premise and conclusion.
  Next, for each variable $v$ that occurs $n > 0$ times in the premise $\mathcal{F}$, we choose a list $v = v^1, v^2, \dots, v^n$ of variables of the same sort, where $v^2, \dots, v^n$ are fresh.
  We now replace the $i$th occurrence of $v$ in $\mathcal{F}$ with $v^i$, and add the atoms $\mathrm{Eq}(v, v^i)$ for all $i = 2, \dots, n$ to $\mathcal{F}$.

  The transformed premises $\mathcal{F}''$ have the following property:
  If $X$ is a relational $\mathfrak{S}'$-structure that satisfies the equivalence relation axioms and $X'$ is the relational structure over $X$ that furthermore satisfies the congruence axioms, then maps $[\mathcal{F}'] \rightarrow X'$ are in bijection to maps $[\mathcal{F}''] \rightarrow X$ up to setoid morphism equality.
  From this it follows that Corollary~\ref{cor:setoid-model-free} holds also for the sparse transformation $(\mathfrak{S}', T'')$.

  The sparse transformation avoids duplication in many cases, but data that can be inferred twice for distinct but equivalent elements is still duplicated.
  Furthermore, the transformed premises $\mathcal{F}''$ can be more computationally expensive to match because all elements in an equivalence class must be considered for every occurrence of a variable in the original premise $\mathcal{F}$.
\end{remark}

\section{Partial Horn Logic}
\label{sec:phl}

Partial Horn logic is one of the many equivalent notions of essentially algebraic theory.
It was initially defined by Palmgren and Vickers \citep{phl}, who proved its equivalence to essentially algebraic theories.
Here we shall understand PHL as a syntactic extension, as \emph{syntactic sugar}, over RHL.

Observe that it cannot be read off from the individual sequents whether or not an RHL theory is strong or not.
Instead, one has to consider the interplay between the different sequents of the theory.
As a result, it is computationally undecidable whether or not a given RHL theory is strong.

Our application for the semantics developed in this paper are tools that allow computations based on the small object argument.
Such tools compute (fragments of) free models of user-defined theories that encode problem domains.
It is highly desirable that these theories are strong, since otherwise the result of the computation is not uniquely determined.
As strong RHL theories are difficult to recognize for both humans and computers, we argue that RHL is not directly suitable as an input theory language for this purpose.
What is needed, then, is a language with the same descriptive power as RHL, but where an easily recognizable subset allows axiomatizing all strong theories.

PHL is indeed such a language:
Proposition~\ref{prop:strong-rhl-to-epic-phl} shows that every strong RHL theory is equivalent to a PHL theory containing \emph{epic} sequents only.
PHL sequents are epic if no new variables are introduced in the conclusion, which is a criterion that can be easily checked separately for each sequent without regard for the theory the sequent appears in.
If tools wish to allow only strong theories, they can accept PHL as input language but reject non-epic sequents.
Note that there exist PHL theories containing non-epic sequents which are nevertheless strong, but such theories can be equivalently axiomatized as epic PHL theories.
Thus, no generality is lost compared to general strong theories when rejecting non-epic PHL theories.

\subsection{Algebraic Structures}

\begin{definition}
  An \emph{algebraic signature} is a relational signature $(S, R)$ equipped with a partition $R = P \sqcup F$ of the set of relation symbol into disjoint sets $P$ of \emph{predicate symbols} and $F$ of \emph{function symbols} such that the arity of every function symbol is non-empty.
  If $f \in F$ is a function symbol, then we write $f : s_1 \times \dots \times s_n \rightarrow s$ if the arity of $f$ as a relation symbol is $f : s_1 \times \dots \times s_n \times s$.
\end{definition}

\begin{definition}
  An \emph{algebraic $\mathfrak{S}$-structure} is a relational $\mathfrak{S}$-structure $X$ such that $f_X$ is the graph of a partial function for all $f \in F$.
  Thus if $(x_1, \dots, x_n, y) \in f_X$ and $(x_1, \dots, x_n, z) \in f_X$, then $y = z$.
  We use $f_X(x_1, \dots, x_n)$ to denote the unique element $y$ such that $(x_1, \dots, x_n, y) \in f_X$, and we write $f_X(x_1, \dots, x_n) \downarrow$ to denote that such an element $y$ exists.
  A \emph{morphism of algebraic structures} is a morphism of underlying relational structures.
  The \emph{category of algebraic structures} is denoted by $\mathrm{Alg}(\mathfrak{S})$.
\end{definition}
If the algebraic signature $\mathfrak{S}$ is clear from context, we abbreviate $\mathrm{Alg}(\mathfrak{S})$ as $\mathrm{Alg}$.

\begin{proposition}
  \label{prop:functionality-axioms}
  Let $X$ be a relational structure.
  Then $X$ is an algebraic structure if and only if it satisfies the RHL sequent
  \begin{equation}
    \label{eq:functionality-axiom}
    f(v_1, \dots, v_n, u_0) \land f(v_1, \dots, v_n, u_1) \implies u_0 \equiv u_1
  \end{equation}
  for each function symbol $f : s_1 \times \dots \times s_n \rightarrow s$.
  \qed
\end{proposition}

\begin{corollary}
  The category of algebraic structures is a reflective subcategory of the category of relational structures.
  The reflections $X \rightarrow X'$ of relational structures $X$ into $\mathrm{Alg}$ are surjections.
\end{corollary}
\begin{proof}
  That every reflection is surjective follows from the small object argument and the fact that the classifying morphisms of functionality axioms \eqref{eq:functionality-axiom} are epimorphisms of relational structures, hence so are all coproducts, pushouts and (infinite) compositions thereof.
\end{proof}
We denote the free algebraic structure functor by $\mathrm{FAlg} : \mathrm{Rel} \rightarrow \mathrm{Alg}$.

\begin{corollary}
  The category of algebraic structures is complete and cocomplete.
\end{corollary}
\begin{proof}
  This follows from general facts about reflective subcatgories:
  They are stable under limits, and colimits are computed by reflecting colimits of the ambient category.
\end{proof}

\subsection{Syntax and Semantics}
\label{subsec:phl-semantics}

Fix an algebraic signature $\mathfrak{S} = (S, P \sqcup F)$.

\begin{definition}
  \label{def:phl}
  The set of \emph{terms} and a sort assigned to each term is given by the following recursive definition:
  \begin{enumerate}
    \item
      If $v$ is a variable of sort $s$, then $v$ is a term of sort $s$.
    \item
      If $f : s_1 \times \dots \times s_n \rightarrow s$ is a function symbol and $t_1, \dots, t_n$ are terms such that $t_i$ has sort $s_i$ for all $i = 1, \dots, n$, then $f(t_1, \dots, t_n)$ is a term of sort $s$.
  \end{enumerate}
  \emph{PHL atoms, PHL formulas and PHL sequents} are defined as in Definition~\ref{def:relational-syntax}, but with three changes:
  \begin{enumerate}
    \item
      In each type of atom, also composite terms of the same sort are allowed in place of only variables.
    \item
      A PHL atom $r(t_1, \dots, t_n)$ is valid only if $r = p$ is a predicate symbol, but not if $r$ is a function symbol.
      We refer to such atoms as \emph{predicate atoms}.
    \item
      An atom of the form $t \downarrow$ is called a \emph{term assertion atom}, whereas \emph{sort quantification atom} is reserved for atoms of the form $v \downarrow$ with $v$ a variable.
  \end{enumerate}
\end{definition}

\begin{definition}
  Let $X$ be an algebraic structure.
  An \emph{interpretation of a term $t$} in $X$ is an interpretation $I$ of the variables occurring in $t$ such that the following recursive extension of $I$ to composite terms is well-defined on $t$:
  \begin{equation}
    I(f(t_1, \dots, t_n)) = f_X(I(t_1), \dots, I(t_n)).
  \end{equation}
  Note that the right-hand side might not be defined; in this case also the left-hand side is undefined.

  An \emph{interpretation of a PHL atom} in $X$ is defined analogously to the interpretation of an RHL atom, but with the additional condition that the interpretation is defined on all (possibly composite) terms occurring in the atom.

  An \emph{interpretation of a PHL formula $\mathcal{F} = \phi_1 \land \dots \land \phi_n$} is an interpretation of the variables occurring in $\mathcal{F}$ that restricts to an interpretation of $\phi_i$ for each $i \in \{1, \dots, n\}$.
  An algebraic structure $X$ \emph{satisfies} a PHL sequent $\mathcal{F} \implies \mathcal{G}$ if each interpretation of $\mathcal{F}$ in $X$ can be extended to an interpretation of $\mathcal{F} \land \mathcal{G}$ in $X$.
\end{definition}

\begin{remark}
  There are some differences between our notion of partial Horn logic and the notion introduced by \citet{phl}:
  \begin{enumerate}
    \item
      \label{itm:difference-downarrow}
      In \citet{phl}, the notation $t \downarrow$ is a meta-theoretic abbreviation for self-equality atoms $t \equiv t$.
    \item
      \label{itm:difference-context}
      In \citet{phl}, sequents are annotated with a \emph{context}, i.e. a set of variables.
      Sequents may only refer to variables listed in the context.
      Interpretations of the premise of the sequent must always also interpret variables in the context.
      A sequent $\mathcal{F} \implies \mathcal{G}$ with context $C = \{v_1, \dots, v_n\}$ as in \citet{phl} is thus semantically equivalent to our PHL sequent $\mathcal{F} \land v_1 \downarrow \land \dots \land v_n \downarrow \implies \mathcal{G}$.
    \item
      \label{itm:difference-existentials}
      Our notion of sequent allows variables in the conclusion that do not occur in the premise (or context).
      According to our semantics, we may think of such variables as implicitly existentially quantified:
      A sequent is satisfied if for every interpretation of the premise, there exists a suitable extension of the interpretation to the conclusion.
      This does not have a counterpart in the PHL considered in \citet{phl}.
  \end{enumerate}
  \ref{itm:difference-downarrow} and \ref{itm:difference-context} are minor syntactic differences, but \ref{itm:difference-existentials} increases the descriptive strength of PHL non-trivially.
  We later consider \emph{epic} PHL (Section~\ref{subsec:phl-completeness}), where all variables in the conclusion must also occur in the premise.
  Apart from \ref{itm:difference-downarrow} and \ref{itm:difference-context} and above, our notion of epic PHL agrees with PHL as defined in \citet{phl}.
\end{remark}

\begin{definition}
  \label{def:flattening}
  Let $t$ be a term.
  The \emph{flattening} of $t$ consists of an RHL formula $\mathrm{Flat}(t)$ and a result variable $v_\mathrm{Flat}(t)$.
  Flattening is defined recursively as follows:
  \begin{enumerate}
    \item
      If $t = v$ is a variable, then $\mathrm{Flat}(t) = \top$ is the empty conjunction and $v_\mathrm{Flat}(t) = v$.
    \item
      \label{itm:flattening-composite}
      If $t = f(t_1, \dots, t_n)$, then
      \begin{equation}
        \mathrm{Flat}(t) = \mathrm{Flat}(t_1) \land \dots \land \mathrm{Flat}(t_n) \land f(v_\mathrm{Flat}(t_1), \dots, v_\mathrm{Flat}(t_n), u)
      \end{equation}
      where $u = v_\mathrm{Flat}(t)$ is a fresh variable.
  \end{enumerate}

  Let $\phi$ be a PHL atom.
  The flattening $\mathrm{Flat}(\phi)$ is an RHL formula which is defined depending on the type of $\phi$ as follows:
  \begin{enumerate}
    \item
      If $\phi = p(t_1, \dots, t_n)$ is a predicate atom, then
      \begin{equation}
        \mathrm{Flat}(\phi) = \mathrm{Flat}(t_1) \land \dots \land \mathrm{Flat}(t_n) \land p(v_\mathrm{Flat}(t_1), \dots, v_\mathrm{Flat}(t_n)).
      \end{equation}
    \item
      If $\phi = t \downarrow$ is a term assertion atom, then
      \begin{equation}
        \mathrm{Flat}(\phi) = \mathrm{Flat}(t) \land v_\mathrm{Flat}(t) \downarrow.
      \end{equation}
    \item
      If $\phi = t_1 \equiv t_2$ is an equality atom, then
      \begin{equation}
        \mathrm{Flat}(\phi) = \mathrm{Flat}(t_1) \land \mathrm{Flat}(t_2) \land v_\mathrm{Flat}(t_1) \equiv v_\mathrm{Flat}(t_2).
      \end{equation}
  \end{enumerate}
  The flattening of a PHL formula is the conjunction of the flattenings of each atom making up the formula.
  The flattening of a PHL sequent is the RHL sequent given by flattening premise and conclusion.
\end{definition}

\begin{remark}
  \label{rem:flattening}
  The flattening of a composite term $t = f(t_1, \dots, t_n)$ involves the choice of a ``fresh'' variable $u$.
  This notion can be made precise as follows:
  The flattening operations take as additional parameter a sequence $(u_n)_{n \in \mathbb{N}}$ of variables such that the variables $u_i$ do not occur in the syntactic objects that should be flattened.
  Choosing a ``fresh'' variable $u$ now means that we set $u = u_0$, and for all further flattening operations we pass the sequence $(u_{n + 1})_{n \in \mathbb{N}}$.

  This means that a term $t$ that appears twice in the same formula $\mathcal{F}$ will be flattened twice with different choices of fresh variables.
  For example, if $f$ is a binary function symbol and $x_1, x_2$ are variables, then the flattening of the PHL formula
  \begin{equation}
    \mathcal{F} = f(x_1, x_2) \equiv x_1 \land f(x_1, x_2) \equiv x_2
  \end{equation}
  is the RHL formula
  \begin{equation}
    \mathrm{Flat}(\mathcal{F}) = f(x_1, x_2, u_0) \land u_0 \equiv x_1 \land f(x_1, x_2, u_1) \equiv x_2
  \end{equation}
  where $u_0 \neq u_1$.
\end{remark}

\begin{proposition}
  \label{prop:flattening-versus-interpretation}
  Let $X$ be an algebraic structure.
  \begin{enumerate}
    \item
      Let $t$ be a term and let $I$ be an interpretation of the variables of $t$ in $X$.
      Then $I$ can be extended to the term $t$ if and only if $I$ can be extended to an interpretation of the RHL formula $\mathrm{Flat}(t)$.
      In either case, if an extension $J$ exists, then it exists uniquely, and $J(v_\mathrm{Flat}(t)) = I(t)$.
    \item
      Let $\phi$ be a PHL atom and let $I$ be an interpretation of the variables of $\phi$ in $X$.
      Then $I$ is an interpretation of $\phi$ if and only if $I$ can be extended to an interpretation $J$ of the RHL formula $\mathrm{Flat}(\phi)$.
      If $J$ exists, then it exists uniquely.
    \item
      Let $\mathcal{S}$ be a PHL sequent.
      Then $X$ satisfies $\mathcal{S}$ if and only if it satisfies the RHL sequent $\mathrm{Flat}(\mathcal{S})$.
  \end{enumerate}
\end{proposition}
\begin{proof}
  By construction.
\end{proof}

\begin{definition}
  We associate to each PHL formula $\mathcal{F}$ the following \emph{classifying algebraic structure}:
  \begin{equation}
    [\mathcal{F}] = \mathrm{FAlg}([\mathrm{Flat}(\mathcal{F})])
  \end{equation}
  The composition of the generic interpretation $I_{\mathrm{Flat}(\mathcal{F})}$ in $[\mathrm{Flat}(\mathcal{F})]$ with the reflection into $\mathrm{Alg}$ induces an interpretation of $\mathrm{Flat}(\mathcal{F})$ in $[\mathcal{F}]$, which then restricts to an interpretation $I_\mathcal{F}$ of $\mathcal{F}$ in $[\mathcal{F}]$.
  We call $I_\mathcal{F}$ the  \emph{generic interpretation} of $\mathcal{F}$.
\end{definition}

\begin{proposition}
  \label{prop:algebraic-formula-structure-universal}
  Let $\mathcal{F}$ be a PHL formula and let $X$ be an algebraic structure.
  Then there is a bijection between interpretations of $\mathcal{F}$ in $X$ and maps $[\mathcal{F}] \rightarrow X$.
\end{proposition}
\begin{proof}
  This follows by combining Proposition~\ref{prop:flattening-versus-interpretation}, Proposition~\ref{prop:algebraic-formula-structure-universal} and the universal property of the free algebraic structure functor.
\end{proof}

\begin{definition}
  Let $\mathcal{S} = \mathcal{F} \implies \mathcal{G}$ be a PHL sequent.
  The \emph{classifying morphism} of $\mathcal{S}$ is the morphism $[\mathcal{S}] : [\mathcal{F}] \rightarrow [\mathcal{F} \land \mathcal{G}]$ of classifying algebraic structures that is induced by the canonical interpretation of $\mathcal{F}$ in $[\mathcal{F} \land \mathcal{G}]$.
\end{definition}

\begin{proposition}
  Let $\mathcal{S}$ be a PHL sequent and let $X$ be an algebraic structure.
  Then $X$ satisfies $\mathcal{S}$ if and only if $X$ is injective to $[\mathcal{S}]$.
\end{proposition}
\begin{proof}
  Analogous to the proof of Proposition~\ref{prop:relational-sequent-morphism-lifts}.
\end{proof}

\begin{proposition}
  \label{prop:phl-models}
  Let $T$ be a PHL theory.
  Denote the functionality sequent \ref{eq:functionality-axiom} for function symbols $f \in F$ by $f_\mathrm{func}$.
  Then $\mathrm{Mod}(T)$ is equivalent to the following injectivity classes:
  \begin{enumerate}
    \item
      \label{itm:phl-models-in-alg}
      $(M_1)^\pitchfork \subseteq \mathrm{Alg}(\mathfrak{S})$, where $M_1 = \{ [\mathcal{S}] \mid \mathcal{S} \in T \}$.
    \item
      \label{itm:phl-models-in-rel-alg-sequents}
      $(M_2)^\pitchfork \subseteq \mathrm{Rel}(\mathfrak{S})$, where $M_2 = \{ [f_\mathrm{func}] \mid f \in F \} \cup \{ [\mathcal{S}] \mid \mathcal{S} \in T \}$.
    \item
      \label{itm:phl-models-in-rel-rel-sequents}
      $(M_3)^\pitchfork \subseteq \mathrm{Rel}(\mathfrak{S})$, where $M_3 = \{ [f_\mathrm{func}] \mid f \in F \} \cup \{ [\mathrm{Flat}(\mathcal{S})] \mid \mathcal{S} \in T \}$.
  \end{enumerate}
  Here $[\mathcal{S}]$ in the definition of $M_2$ denotes the classifying morphism of the PHL sequent $\mathcal{S}$, which is a morphism of algebraic structures, hence in particular a morphism of relational structures.
  $[\mathrm{Flat}(\mathcal{S})]$ denotes the classifying morphism of the RHL sequent $\mathrm{Flat}(\mathcal{S})$.

  In particular, $\mathrm{Mod}(T)$ is a weakly reflective subcategory of both $\mathrm{Rel}(\mathfrak{S})$ and $\mathrm{Alg}(\mathfrak{S})$.
  If any one of $M_1, M_2$ or $M_3$ are strong, then all of them are strong, and $\mathrm{Mod}(T)$ is a reflective subcategory of $\mathrm{Rel}(\mathfrak{S})$ and $\mathrm{Alg}(\mathfrak{S})$.
\end{proposition}
\begin{proof}
  \ref{itm:phl-models-in-alg} follows from Proposition~\ref{prop:algebraic-formula-structure-universal}, and then \ref{itm:phl-models-in-rel-alg-sequents} and \ref{itm:phl-models-in-rel-rel-sequents} follow from Proposition~\ref{prop:reflection-vs-injectivity}.
\end{proof}

\begin{proposition}
  \label{prop:unflatten}
  Let $\phi$ be an RHL formula.
  Then there exists a PHL formula $\mathrm{Unflat}(\phi)$ with the following properties:
  \begin{enumerate}
    \item
      Every variable occurs in $\phi$ if and only if it occurs in $\mathrm{Unflat}(\phi)$.
    \item
      Let $I$ be an interpretation of the variables of $\phi$ in an algebraic structure $X$.
      Then $I$ is an interpretation of the PHL formula $\phi$ if and only if it is an interpretation of the RHL formula $\mathrm{Unflat}(\phi)$.
  \end{enumerate}
\end{proposition}
\begin{proof}
  Replace every relation atom of the form $f(x_1, \dots, x_n, x)$ for some function symbol $f : s_1 \times \dots \times s_n \rightarrow s$ with the PHL atom $f(x_1, \dots, x_n) \equiv x$.
\end{proof}

\begin{proposition}
  \label{prop:composition-algebraic-sequent-morphisms}
  Let $\mathcal{F}, \mathcal{G}$ and $\mathcal{H}$ be PHL formulas.
  Then
  \begin{equation}
    [\mathcal{F} \land \mathcal{G} \implies \mathcal{H}] \circ [\mathcal{F} \implies \mathcal{G}] \cong [\mathcal{F} \implies \mathcal{G} \land \mathcal{H}].
  \end{equation}
\end{proposition}
\begin{proof}
  This follows from the universal property of classifying morphisms.
\end{proof}

\subsection{Completeness Results}
\label{subsec:phl-completeness}

In this section, we study the descriptive strength of PHL.
As with RHL and relational structures, every map of finite algebraic structures can be described as classifying morphism of a PHL sequent (Proposition~\ref{prop:image-of-classifying-phl-sequents}).
We classify epimorphisms of algebraic structures and show that epimorphisms of finite algebraic structures correspond to \emph{epic} PHL sequents, where variables cannot be introduced in the conclusion (Proposition~\ref{prop:classifying-alg-epic}).
Finally, we show that every strong RHL theory is semantically equivalent to an epic PHL theory (Proposition~\ref{prop:strong-rhl-to-epic-phl}).

\begin{proposition}
  \label{prop:image-of-classifying-phl-sequents}
  Let $\mathcal{S}$ be a PHL sequent.
  Then the classifying morphism $[\mathcal{S}]$ is a morphism of finite algebraic structures.
  Conversely, for every morphism $f : X \rightarrow Y$ of finite algebraic structures, there exists a PHL sequent $\mathcal{S} = \mathcal{F} \implies \mathcal{G}$ such that $f$ and $[\mathcal{S}]$ are isomorphic.
\end{proposition}
\begin{proof}
  By Propositions \ref{prop:image-of-classifying-rhl-sequents} and \ref{prop:unflatten}.
\end{proof}

\begin{definition}
  Let $f$ be a function symbol.
  The \emph{totality sequent} $f \downarrow$ is given by
  \begin{equation}
    v_1 \downarrow \land \dots \land v_n \downarrow \implies f(v_1, \dots, v_n)\downarrow.
  \end{equation}
  We denote by
  \begin{equation}
    \mathrm{Tot} = \{ [f \downarrow] \mid f \in F \}
  \end{equation}
  the set of classifying morphisms of totality sequents.
\end{definition}

\begin{proposition}
  Let $f$ be a function symbol. 
  \begin{enumerate}
    \item
      An algebraic structure $X$ satisfies $f \downarrow$ if and only if $f_X$ is a total function.
    \item
      $[f \downarrow]$ is an epimorphism of algebraic structures.
      \qed
  \end{enumerate}
\end{proposition}

\begin{definition}
  Let $g : X \rightarrow Y$ be  a map of algebraic structures.
  We say that $X$ is \emph{total over Y} (with respect to $g$) if and only if for all function symbols $f$ and elements $x_1, \dots, x_n \in X$, if $f_Y(g(x_1), \dots, g(x_n))$ is defined, then $f_X(x_1, \dots, x_n)$ is defined.
\end{definition}

\begin{proposition}
  \label{prop:saturation}
  Let $g : X \rightarrow Y$ be a map of algebraic structures.
  Then there exists a factorization
  \begin{equation}
    \begin{tikzcd}
      X \arrow[dr, "g"'] \arrow[rr, "h"] & & X' \arrow[dl, "g'"] \\
      & Y
    \end{tikzcd}
  \end{equation}
  such that $h$ is a relative $\mathrm{Tot}$-cell complex and $X'$ is total over $Y$.
  Moreover, the triple $(X', h, g')$ is uniquely determined by $g$ up to unique isomorphism.
\end{proposition}
\begin{proof}
  Consider the set $\mathrm{Tot}_Y$ of all triples $(f, a, b)$ of function symbols $f$ and pairs of morphisms $a$ and $b$ making up commuting triangles
  \begin{equation}
    \begin{tikzcd}
      \cdot \arrow[dr, "a"'] \arrow[rr, "{[f \downarrow]}"] & & \cdot \arrow[dl, "b"] \\
      & Y
    \end{tikzcd}
  \end{equation}
  $\mathrm{Tot}_Y$ is a set of epimorphisms in the slice category $\mathrm{Alg}_{/ Y}$.
  It follows that $\mathrm{Tot}_Y$ is strong, so $(\mathrm{Tot}_Y)^\pitchfork$ is a reflective subcategory of $\mathrm{Alg}_{/ Y}$.
  Partial algebras over $Y$ are total over $Y$ if and only if they are injective to $\mathrm{Tot}_Y$.
  It follows that a triple $(X', h, g')$ with $h$ a relative $\mathrm{Tot}_Y$-complex and $g$ relatively total exists and is unique up to unique isomorphism.
  
  It remains to show that $h$ is a relative $\mathrm{Tot}$-cell complex.
  By definition, the class of relative $M$-cell complexes is obtained from $M$ by closure under certain classes of colimits.
  The forgetful functor $\mathrm{Alg}_{/ Y} \rightarrow \mathrm{Alg}$ preserves colimits and maps $\mathrm{Tot}_Y$ into $\mathrm{Tot}$.
  From this it follows that the image of a relative $\mathrm{Tot}_Y$-cell complex in $\mathrm{Alg}$ is a relative $\mathrm{Tot}$-complex.
  In particular, $h$ is a relative $\mathrm{Tot}$-cell complex.
\end{proof}

\begin{definition}
  Let $g : X \rightarrow Y$ be a map of algebraic structures.
  We call the unique factorization $g = g' h$ as in Proposition~\ref{prop:saturation} the \emph{relative totalization} of $X$.
\end{definition}

\begin{proposition}
  \label{prop:relatively-total-epis}
  Let $g : X \rightarrow Y$ be a map of algebraic structures such that $X$ is total over $Y$.
  Then $g$ is an epimorphism in $\mathrm{Alg}$ if and only if it is surjective.
\end{proposition}
\begin{proof}
  Since $\mathrm{Alg}$ is a subcategory of $\mathrm{Rel}$, every morphism in $\mathrm{Alg}$ which is an epimorphism in $\mathrm{Rel}$ is also an epimorphism in $\mathrm{Alg}$.
  Thus, every surjective morphism of algebraic structures is an epimorphism.

  Conversely, suppose that $g : X \rightarrow Y$ is an epimorphism of algebraic structures such that $X$ is total over $Y$.
  Let $X \twoheadrightarrow \operatorname{Im}^\mathrm{Rel} g \hookrightarrow Y$ be the image factorization of $g$ in $\mathrm{Rel}$.
  Thus, $(\operatorname{Im}^\mathrm{Rel} g)_s = \{ g_s(x) \mid x \in X_s \}$ for all sorts $s$, and $r_{(\operatorname{Im}^\mathrm{Rel} g)} = \{ (g(x_1), \dots, g(x_n)) \mid (x_1, \dots, x_n) \in r_X \}$ for all relations $r$.
  Since $Y$ is an algebraic structure, the relational substructure $\operatorname{Im}^\mathrm{Rel} g$ is an algebraic structure.
  Because $X$ is relatively total over $Y$, also $\operatorname{Im}^\mathrm{Rel}$ is relatively total over $Y$.
  Since $g$ is surjective if and only if $\operatorname{Im}^\mathrm{Rel} \hookrightarrow Y$ is surjective (that is, a bijection on carriers), we may henceforth assume that $g : X = \operatorname{Im}^\mathrm{Rel} g \hookrightarrow Y$ is injective.

  Now consider the pushout $Z = Y \amalg^\mathrm{Rel}_X Y$ in $\mathrm{Rel}$.
  There are inclusions $Y \cong Y_0 \subseteq Z$ and $Y \cong Y_1 \subseteq Z$ corresponding to the two components of $Z$ such that $Y_0 \cup Y_1 = Z$, and inclusions $X \subseteq Y_0, X \subseteq Y_1$.
  We have $(Y_0)_s \cap (Y_1)_s = X_s$ for all sorts $s$, but note that the analogous equation does not necessarily hold for relations.

  We claim that $Z$ is an algebraic structure.
  Thus let $f$ be a function symbol, and let $z_1, \dots, z_n, z, z' \in Z$ such that $\bar z = (z_1, \dots, z_n, z) \in f_Z$ and $\bar z' = (z_1, \dots, z_n, z') \in f_Z$.
  We need to show that $z = z'$.

  If $\bar z, \bar z' \in Y_0$ or $\bar z, \bar z' \in Y_1$ this follows from the fact that $Y_0 \cong Y \cong Y_1$ is an algebraic structure.
  We may thus (by symmetry) assume that $\bar z \in Y_0$ and $\bar z' \in Y_1$.
  Because the first $n$ projections of $\bar z$ and $\bar z'$ agree, we have $z_i \in Y_0 \cap Y_1$, hence $z_i \in X$ for $i \in \{ 1, \dots, n\}$.
  Because $X$ is total over the $Y_i$ with respect to the inclusions $X \subseteq Y_i$, we have $f_X(z_1, \dots, z_n) = z''$ for some $z'' \in X$.
  It follows that $z = z''$ and $z' = z''$, hence $z = z'$.

  As in every cocomplete category, $g$ is an epimorphism of algebraic structures if and only if the two maps $Y \rightarrow Y \amalg^\mathrm{Alg}_X Y \eqqcolon Z'$ to the pushout of algebraic structures agree.
  But we have just shown that $Z' = Z$, so also the two maps $Y \rightarrow Y \amalg^\mathrm{Rel}_X Y$ agree.
  Thus $g$ is an epimorphism in $\mathrm{Rel}$, hence surjective.
\end{proof}

\begin{proposition}
  \label{prop:monic-epic-alg-morphisms}
  Let $g : X \rightarrow Y$ be a map of algebraic structures.
  \begin{enumerate}
    \item
      \label{itm:algebraic-monomorphisms}
      $g$ is a monomorphism in $\mathrm{Alg}$ if and only if it is injective.
    \item
      \label{itm:algebraic-epimorphisms}
      Let $g' h = g$ be the relative totalization of $X$.
      Then $g$ is an epimorphism in $\mathrm{Alg}$ if and only if $g'$ is a surjection.
  \end{enumerate}
\end{proposition}
\begin{proof}
  \ref{itm:algebraic-monomorphisms}.
  In general, morphisms in reflective subcategories are monic if and only if they are monic as morphisms in the ambient category.

  \ref{itm:algebraic-epimorphisms}.
  Every morphism in $\mathrm{Tot}$ is an epimorphism.
  Since epimorphisms are stable under pushouts and infinite compositions, every relative $\mathrm{Tot}$-cell complex and in particular $h$ is an epimorphism.
  Thus $g$ is an epimorphism in $\mathrm{Alg}$ if and only if $g'$ is an epimorphism in $\mathrm{Alg}$.
  We conclude with Proposition~\ref{prop:relatively-total-epis}.
\end{proof}

\begin{definition}
  \label{def:epic-phl-sequents}
  A PHL sequent $\mathcal{S} = \mathcal{F} \implies \mathcal{G}$ is \emph{epic} if every variable in the conclusion $\mathcal{G}$ also occurs in the premise $\mathcal{F}$.
\end{definition}

\begin{proposition}
  \label{prop:flattening-atoms}
  Let $\mathcal{F}$ be a PHL formula.
  Let $\mathrm{Flat}(\mathcal{F}) = \phi_1 \land \dots \land \phi_n$ for RHL atoms $\phi_i$.
  Let $\phi = \phi_i$ for some $i \in \{1, \dots, n\}$ and let
  \begin{equation}
    V = \{ v \mid v \text{ occurs in } \mathcal{F} \text{ or in } \phi_j \text{ for some } j < i \}.
  \end{equation}
  Then $\phi$ has one of the following forms:
  \begin{enumerate}
    \item
      \label{itm:flattened-predicate-case}
      $\phi = p(v_1, \dots, v_m)$, where $p$ is a predicate symbol and $v_1, \dots, v_m \in V$.
    \item
      \label{itm:flattened-function-case}
      $\phi = f(v_1, \dots, v_m, v)$, where $f : s_1 \times \dots \times s_m \rightarrow s$ is a function symbol, $v_1, \dots, v_m \in V$ and $v$ is a fresh variable.
    \item
      \label{itm:flattened-defined-case}
      $\phi = v \downarrow$, where $v \in V$.
    \item
      \label{itm:flattened-equals-case}
      $\phi = v_1 \equiv v_2$, where $v_1, v_2 \in V$.
  \end{enumerate}
\end{proposition}
\begin{proof}
  Follows inductively from the definition of flattening; see also Remark~\ref{rem:flattening}.
\end{proof}

\begin{proposition}
  \label{prop:epic-stable-under-composition}
  The classifying morphisms of epic PHL sequents are, up to isomorphism, closed under composition.
\end{proposition}
\begin{proof}
  Let $\mathcal{S} = \mathcal{F} \implies \mathcal{G}$ and $\mathcal{T} = \mathcal{G}' \implies \mathcal{H}$ be epic PHL sequents and suppose that there exists an isomorphism $k : [\mathcal{G}] \cong [\mathcal{G}']$.
  It suffices to show that $[\mathcal{T}] \circ k$ is isomorphic to the classifying morphism of an epic PHL sequent of the form $\mathcal{G} \implies \mathcal{H}'$, since then $[\mathcal{F} \implies \mathcal{G} \land \mathcal{H}'] \cong [\mathcal{T}] \circ k \circ [\mathcal{S}]$ by Proposition~\ref{prop:composition-algebraic-sequent-morphisms}.

  Every variable $v$ in $\mathcal{G}'$ corresponds to some term $t$ in $\mathcal{G}$ under $k$.
  More precisely, the canonical interpretation of $v$ and the inverse of $k$ determine an element $x = k^{-1}(I_{\mathcal{G}'}(v)) \in [\mathcal{G}]$, and then $x$ is the canonical interpretation $x = I_{\mathcal{G}}(t)$ of some term $t$ that occurs in $\mathcal{G}$.
  Let $\mathcal{H}'$ be the formula that is obtained from $\mathcal{H}$ by replacing every variable $v$ by a corresponding term $t$ in $\mathcal{G}$.
  Observe that we assumed $\mathcal{G}' \implies \mathcal{H}$ to be an epic sequent, so every variable in $\mathcal{H}$ also occurs in $\mathcal{G}'$.
  Thus, $\mathcal{G} \implies \mathcal{H}'$ is epic.

  By induction over the number of atoms in $\mathcal{H}$, we find an interpretation of $\mathcal{G} \land \mathcal{H}'$ in $[\mathcal{G}' \land \mathcal{H}]$ that is compatible with $k$, and vice versa there is an interpretation of $\mathcal{G}' \land \mathcal{H}$ in $[\mathcal{G} \land \mathcal{H}']$ that is compatible with $k^{-1}$.
  These interpretations then induce a commutative square
  \begin{equation}
    \begin{tikzcd}
      \left[\mathcal{G}\right] \arrow[r] \arrow[d, "k"', "\cong"] & \left[\mathcal{G} \land \mathcal{H}'\right] \arrow[d, "\cong"] \\
      \left[\mathcal{G}'\right] \arrow[r] & \left[\mathcal{G}' \land \mathcal{H}\right]
    \end{tikzcd}
  \end{equation}
  as desired.
\end{proof}

\begin{proposition}
  \label{prop:classifying-alg-epic}
  The classifying morphisms of epic PHL sequents are, up to isomorphism, precisely the epimorphisms of finite algebraic structures.
\end{proposition}
\begin{proof}
  Let $\mathrm{Flat}(\mathcal{S}) = \mathcal{F} \implies \mathcal{G}$ be the flattening of an epic PHL sequent $\mathcal{S}$.
  Then $g = \mathrm{FAlg}([\mathcal{F} \implies \mathcal{G}])$ is the classifying morphism of $\mathcal{S}$.
  We need to show that $g$ is an epimorphism.
  Since epimorphisms are stable under composition, we may assume that $\mathcal{G} = \phi$ is a single RHL atom of one of the types listed in Proposition~\ref{prop:flattening-atoms} where $V$ is the set of variables occurring in $\mathcal{F}$.

  If $\phi$ is an atom of type \ref{itm:flattened-predicate-case}, \ref{itm:flattened-defined-case} or \ref{itm:flattened-equals-case}, then the morphism $[\mathcal{F} \implies \phi]$ of relational structures is surjective.
  Because the free algebraic structure functor preserves epimorphisms, this implies that $g = \mathrm{FAlg}([\mathcal{F} \implies \phi])$ is an epimorphism in $\mathrm{Alg}$.
  In case \ref{itm:flattened-function-case}, $[\mathcal{F} \implies \phi]$ is a pushout of $[f \downarrow]$ in $\mathrm{Rel}$.
  Since $[f \downarrow] = \mathrm{FAlg}([f \downarrow])$ is an epimorphism in $\mathrm{Alg}$ and pushouts preserve epimorphisms, it follows that also in this case $g$ is an epimorphism.

  Now suppose that $g : X \rightarrow Y$ is an epimorphism of finite algebraic structures.
  Let $g = g' h$ be the relative totalization of $X$ over $Y$.
  Since $h$ is a relative $\mathrm{Tot}$-cell complex, there exists by Proposition~\ref{prop:iterative-relative-cell-complexes} a sequence of pushout squares
  \begin{equation}
    \begin{tikzcd}
      \left[ v_1 \downarrow \land \dots \land v_n \downarrow \right] \arrow[r, "{[f_n \downarrow]}"] \arrow[d, "a"] & \left[ f_n(v_1, \dots, v_n) \downarrow \right] \arrow[d, "b"] \\
      X_n \arrow[r, "h_n"] & X_{n + 1}
    \end{tikzcd}
  \end{equation}
  for a totality sequent $f_n \downarrow$ for all $n$ such that $h$ is the infinite composition of the $h_n$.
  Note that, a priori, Proposition~\ref{prop:iterative-relative-cell-complexes} implies only that $h_n$ is a pushout of a \emph{coproduct} of totality sequents.
  However, finiteness of $X$ and $Y$ implies inductively that these coproducts can be chosen to be finite, and then a single pushout of a finite coproduct can equivalently be written as a finite composition of pushouts.

  We claim that $h_n \cong [\mathcal{S}_n]$ for all $n$ and a sequence of epic PHL sequents $(\mathcal{S}_n)_{n \in \mathbb{N}}$.
  To verify this, choose first a PHL formula $\mathcal{F}$ such that $[\mathcal{F}] \cong X_n$.
  A formula $\mathcal{F}$ with this property exists by Proposition~\ref{prop:image-of-classifying-phl-sequents} because the identity on $X_n$ is a map of finite algebraic structures.
  The map $a$ corresponds to elements $x_1, \dots, x_n \in X$, and these elements are the interpretations of terms $t_1, \dots, t_n$ that occur in $\mathcal{F}$.
  Now
  \begin{equation}
    h_n \cong [\mathcal{F} \implies f_n(t_1, \dots, t_n)\downarrow].
  \end{equation}

  Let $g'_n : X_n \rightarrow \operatorname{colim}_i X_i \xrightarrow{g'} Y$ for $n \geq 0$.
  Because $Y$ is finite, the chain
  \begin{equation}
    \operatorname{Im} g'_0 \subseteq \operatorname{Im} g'_1 \subseteq \dots \subseteq Y
  \end{equation}
  is eventually stationary, say for $n \geq n_0$, and then $\operatorname{Im} g'_{n_0} = \operatorname{Im} g'$.
  $g$ is an epimorphism, hence $g'$ is a surjection by Proposition~\ref{prop:monic-epic-alg-morphisms}, hence also $g'_{n_0} : X_{n_0} \rightarrow Y$ is a surjection.
  Thus $g'_{n_0} \cong [\mathrm{Unflat}(\mathcal{S})]$ for some surjective RHL sequent $\mathcal{S}$.
  Note that the PHL sequent $\mathrm{Unflat}(\mathcal{S})$ is epic because the RHL sequent $\mathcal{S}$ is surjective.

  We have thus decomposed $g$ into a composition
  \begin{equation}
    X = X_0 \xrightarrow{h_0} X_1 \xrightarrow{h_1} X_2 \xrightarrow{h_2} \dots \xrightarrow{h_{n_0 - 1}} X_{n_0} \xrightarrow{g_{n_0}} Y
  \end{equation}
  in which each map is isomorphic to the classifying morphism of an epic PHL sequent $\mathcal{S}$.
  Thus by Proposition~\ref{prop:epic-stable-under-composition}, $g \cong [\mathcal{S}]$ for some epic PHL sequent $\mathcal{S}$.
\end{proof}

\begin{corollary}
  \label{cor:epic-phl-has-strong-models}
  Let $T'$ be an epic PHL theory, i.e. a PHL theory comprised of epic sequents only.
  Then $\mathrm{Mod}(T')$ is a reflective subcategory of both $\mathrm{Alg}$ and $\mathrm{Rel}$.
\end{corollary}
\begin{proof}
  By combining Proposition~\ref{prop:phl-models} with Proposition~\ref{prop:classifying-alg-epic}.
\end{proof}

\begin{remark}
  There appears to be no simple analogue to Proposition~\ref{prop:classifying-alg-epic} that classifies PHL sequents corresponding to monomorphisms of algebraic structures.
  One possible notion of monic PHL sequent one might consider are PHL sequents $\mathcal{S} = \mathcal{F} \implies \mathcal{G}$ in which the conclusion $\mathcal{G}$ does not contain equality atoms.
  Indeed, the classifying morphisms of such sequents are monic.
  But not every monomorphism of finite algebraic structures can be described as classifying morphism of such sequents.

  For example, consider a signature given by a single sort $s$ and a nullary function symbol $c : s$, and the inclusion $g : X \hookrightarrow Y$ of algebraic structures given by singleton carriers $X_s = Y_s = \{ * \}$ and $c_X = \emptyset, c_Y = \{ * \}$.
  The RHL sequent corresponding to this inclusion is $\mathcal{S} = v \downarrow \implies c(v)$.
  But since $c$ is a function symbol, $c(v)$ is not a valid PHL atom.
  Indeed, unflattening $\mathcal{S}$ yields the PHL sequent $v \downarrow \implies c() \equiv v$, which would not be monic per our proposed definition.

  To rectify this, we might then try to change the definition of PHL so that relation atoms $r(v_1, \dots, v_n, v_{n + 1})$ are valid also in case $r = f : s_1 \times \dots \times s_n \rightarrow s$ is a function symbol.
  But now with the notion of monic PHL sequent considered above, not all classifying morphisms of monic PHL sequents would be monomorphisms, for example for the sequent $v \downarrow \land \; c(u) \implies c(v)$.
\end{remark}

\begin{proposition}
  \label{prop:strong-rhl-to-epic-phl}
  Let $\mathfrak{S} = (S, R)$ be a relational signature and let $T$ be an RHL theory for $\mathfrak{S}$.
  Then there exists an algebraic signature $\mathfrak{S}' = (S, P \sqcup F)$ on the same set of sorts $S$ such that $P = R$ and an epic PHL theory $T'$ for $\mathfrak{S}'$ such that the forgetful functor $\mathrm{Alg}(\mathfrak{S}') \rightarrow \mathrm{Rel}(\mathfrak{S})$ restricts to an equivalence $\mathrm{Mod}(T') \simeq \{ [\mathcal{S}] \mid \mathcal{S} \in T \}^\perp$.
  In particular, if $T$ is strong, then $\mathrm{Mod}(T') \simeq \mathrm{Mod}(T)$.
\end{proposition}
\begin{proof}
  Observe that a relational $\mathfrak{S}$-structure is orthogonal to a classifying morphism $[\mathcal{S}]$ of an RHL sequent $\mathcal{S} = \mathcal{F} \implies \mathcal{G}$ if every interpretation of the premise $\mathcal{F}$ extends \emph{uniquely} to an interpretion of $\mathcal{G}$.
  Our strategy is to add function symbols for each variable in a conclusion of a sequent in $T$ that does not occur in the premise.
  We then add axioms enforcing that each conclusion variable can be obtained by application of the corresponding function symbol to the variables in the premise and vice versa.
  Because evaluation of partial functions yields a unique result, this enforces that the interpretion of conclusion variables is uniquely determined by the interpretation of premise variables.

  In detail, our set of function symbols $F$ is given by
  \begin{equation}
    F = \{ f_{\mathcal{S}, v} \mid \mathcal{S} = \mathcal{F} \implies \mathcal{G} \text{ is in } T, v \text{ occurs in } \mathcal{G} \text{ but not in } \mathcal{F} \}.
  \end{equation}
  Let $\mathcal{S} = \mathcal{F} \implies \mathcal{G}$ be in $T$.
  Let $v_1, \dots, v_n$ be an enumeration of the variables in $\mathcal{F}$ with sorts $s_1, \dots, s_n$.
  Let $v$ be a variable in $\mathcal{G}$ that does not occur in $\mathcal{F}$, and let $s$ be the sort of $v$.
  Then the signature of $f_{\mathcal{S}, v}$ is given by $f_{\mathcal{S}, v} : s_1 \times \dots \times s_n \rightarrow s$.

  Note that each relation symbol in $\mathfrak{S}$ corresponds to a predicate symbol in $\mathfrak{S}'$.
  We thus implicitly coerce RHL sequents for $\mathfrak{S}$ to PHL sequents for $\mathfrak{S}'$ (without invoking $\mathrm{Unflat}$).

  Let $T'$ be the set containing the following PHL sequents, for all $\mathcal{S} = \mathcal{F} \implies \mathcal{G}$ in $T$:
  \begin{enumerate}
    \item
      \label{itm:new-functions-satisfy}
      The sequent $\mathcal{F} \implies \mathcal{G}'$, where $\mathcal{G}'$ is obtained from $\mathcal{G}$ by replacing each variable $v$ in $\mathcal{G}$ that does not occur in $\mathcal{F}$ with $f_{\mathcal{S}, v}(v_1, \dots, v_n)$.
    \item
      \label{itm:new-functions-only-defined-if-premise}
      The sequents
      \begin{equation}
        f_{\mathcal{S}, v}(v_1, \dots, v_n) \downarrow \implies \mathcal{F}
      \end{equation}
      for all variables $v$ in $\mathcal{G}$ that do not occur in $\mathcal{F}$.
    \item
      \label{itm:new-functions-uniquely}
      The sequent
      \begin{equation}
        \mathcal{F} \land \mathcal{G} \implies \bigwedge_{v} v \equiv f_{\mathcal{S}, v}(v_1, \dots, v_n).
      \end{equation}
      where $v$ ranges over the variables in $\mathcal{G}$ that do not occur in $\mathcal{F}$.
  \end{enumerate}
  Clearly all PHL sequents in $T'$ are epic.
  Let $G : \mathrm{Alg}(\mathfrak{S})' \rightarrow \mathrm{Rel}(\mathfrak{S})$ be the forgetful functor.

  We first show that if $X \in \mathrm{Mod}(T')$, then $G(X) \perp [\mathcal{S}]$ for every sequent $\mathcal{S} \in T$.
  Let $V = \{ v_1, \dots, v_n \}$ be the enumeration of variables in $\mathcal{F}$ that we chose in the definition of the signature of the function symbols $f_{\mathcal{S}, v}$.
  Let $I$ be an interpretation of $\mathcal{F}$ in $G(X)$.
  As mentioned earlier, we implicitly treat $\mathcal{F}$ also as a PHL sequent for the signature $\mathfrak{S}'$, and under this identification we can view $I$ as an interpretation of the PHL sequent $\mathcal{F}$ in the algebraic structure $X$.
  Because $X$ satisfies sequent \ref{itm:new-functions-satisfy}, it follows that $I$ is also an interpretation of $\mathcal{G}'$ in $X$.
  By definition of $\mathcal{G}'$, it follows that we obtain an interpretation $J$ of $\mathcal{F} \land \mathcal{G}$ by setting $J(v) = I(v)$ if $v$ occurs in $\mathcal{F}$ and $J(v) = I(f_{\mathcal{S}, v}(v_1, \dots, v_n))$ if $v$ does not occur in $\mathcal{F}$.

  Thus $G(X)$ satisfies the sequent $\mathcal{F} \implies \mathcal{G}$.
  Two interpretations of $\mathcal{F} \land \mathcal{G}$ in $G(X)$ that agree on the variables in $\mathcal{F}$ agree also on the variables that occur in $\mathcal{G}$ because of sequent \ref{itm:new-functions-uniquely}.

  Next we construct a model $X = H(Y) \in \mathrm{Mod}(T')$ given $Y \in \mathrm{Rel}(\mathfrak{S})$ such that $Y \perp [\mathcal{S}]$ for $\mathcal{S} \in T$.
  Set $X_s = Y_s$ for all sorts $s \in S$ and $r_X = r_Y$ for $r \in R$.
  Let $\mathcal{S} = \mathcal{F} \implies \mathcal{G}$ be in $T$, and let $v_1, \dots, v_n$ be the enumeration of the variables in $\mathcal{F}$ that we chose earlier.
  Then we set
  \begin{equation}
    \label{eq:definition-epic-transformation-functions}
    (f_{\mathcal{S}, v})_X(J(v_1), \dots, J(v_n)) = J(v)
  \end{equation}
  whenever $J$ is an interpretation of $\mathcal{F} \land \mathcal{G}$ in $Y$.
  Since $Y$ is orthogonal to $[\mathcal{F} \implies \mathcal{G}]$, two interpretations of $\mathcal{F} \land \mathcal{G}$ are equal as soon as the interpretations agree on the variables $v_1, \dots, v_n$ of $\mathcal{F}$.
  Thus, equation \eqref{eq:definition-epic-transformation-functions} yields well-defined partial functions $(f_{\mathcal{S}, v})_X$.
  By construction, $X$ satisfies sequent \ref{itm:new-functions-satisfy} and the sequents \ref{itm:new-functions-only-defined-if-premise}.
  Satisfaction of sequent \ref{itm:new-functions-uniquely} follows again from uniqueness of the interpretation of conclusion variables.

  The assignment $Y \mapsto H(X)$ is functorial.
  Indeed, let $g : Y_0 \rightarrow Y_1$ be a map in $\mathrm{Rel}(\mathfrak{S})$ with $Y_0$ and $Y_1$ orthogonal to $[\mathcal{S}]$ for $\mathcal{S} \in T$, and let $X_i = H(Y_i)$ for $i \in \{0, 1\}$.
  We need to show that the action of $g$ on carrier sets is also a morphism $X_0 \rightarrow X_1$, i.e. that it preserves partial functions.
  This follows from the definition of the partial functions \eqref{eq:definition-epic-transformation-functions} because if $J$ is an interpretation of $\mathcal{F} \land \mathcal{G}$ in $Y_0$, then $g \circ J$ is an interpretation of $\mathcal{F} \land \mathcal{G}$ in $Y_1$.

  Clearly $G \circ H = \mathrm{Id}$ is the identity functor.
  If $X \in \mathrm{Mod}(T')$, then the partial functions of $X$ satisfy equation \eqref{eq:definition-epic-transformation-functions}.
  Thus, $H \circ G = \mathrm{Id}$ is the identity functor.
\end{proof}

\bibliographystyle{abbrvnat}
\setcitestyle{authoryear,open={(},close={)}}
\bibliography{main}

\end{document}